\documentclass{lmcs} 

\usepackage{lastpage}
\renewcommand{\dOi}{14(4:13)2018}
\lmcsheading{}{1--\pageref{LastPage}}{}{}%
{Nov.~18,~2017}{Nov.~15,~2018}{}


\usepackage{hyperref}

\usepackage{url}

\usepackage{amssymb}
\usepackage{amsmath}
\usepackage{stmaryrd}

\usepackage{tikz}
\usetikzlibrary{calc,positioning,shapes,arrows,3d,fadings,decorations.text, intersections}
\tikzset{every picture/.style={>=angle 60}}
\tikzstyle{ran}=[rounded corners,thick,draw,minimum size=1.4em,inner sep=.5ex]
\tikzstyle{tran}=[thick,draw,->]

\usepackage{macros}

\begin{document}

\title[Game Characterization of Probabilistic Bisimilarity]{Game Characterization of Probabilistic Bisimilarity, and Applications to Pushdown Automata}

\author[V.~Forejt]{Vojt\v{e}ch Forejt\rsuper{1}}	
\address{\lsuper{1}Diffblue Ltd}	
\email{forejtv@gmail.com}  

\author[P.~Jan\v{c}ar]{Petr Jan\v{c}ar\rsuper{2}}	
\address{\lsuper{2}Dept of Computer Science, Faculty of Science,
Palack\'{y} Univ., Olomouc, Czechia}
\email{petr.jancar@upol.cz}  

\author[S.~Kiefer]{Stefan Kiefer\rsuper{3}}	
\address{\lsuper{3}Dept of Computer Science, University of Oxford, UK}
\email{stekie@cs.ox.ac.uk}

\author[J.~Worrell]{James Worrell\rsuper{3}}
\address{\vskip-7pt}
\email{James.Worrell@cs.ox.ac.uk}





\begin{abstract}
\noindent
We study the bisimilarity problem for probabilistic pushdown automata
(pPDA) and subclasses thereof.  Our definition of pPDA allows both
probabilistic and non-deterministic branching, generalising the
classical notion of pushdown automata (without
epsilon-transitions).  
We first show a general 
 characterization of probabilistic bisimilarity in terms of two-player 
 games, which naturally 
reduces checking bisimilarity of probabilistic labelled
transition systems to checking bisimilarity of standard
(non-deterministic)
labelled transition systems. 
This reduction can be easily implemented in the framework of pPDA,
allowing to use known results for standard (non-probabilistic) PDA and
their subclasses. A direct use of the reduction incurs an exponential
increase of complexity, which does not matter in deriving decidability of
bisimilarity for pPDA due to the
non-elementary complexity of the problem.
In the cases of probabilistic one-counter automata
(pOCA),
of probabilistic visibly pushdown automata
(pvPDA), 
and of probabilistic basic process
algebras (i.e., single-state pPDA) we show that an implicit use of
the reduction can avoid the complexity increase; we thus 
get PSPACE, EXPTIME, and 2-EXPTIME upper bounds, respectively, like for
the respective non-probabilistic versions.
The bisimilarity problems for OCA and vPDA are known to have
matching lower bounds (thus being PSPACE-complete and
EXPTIME-complete, respectively); we show that these lower bounds also
hold for fully probabilistic versions that do not use non-determinism.
\end{abstract}

\maketitle


\section{Introduction}
Equivalence checking is the problem of determining whether two systems
are semantically identical.  This is an important question in
automated verification and, more generally, represents a line of
research that can be traced back to the inception of theoretical
computer science.  A great deal of work in this area has been devoted
to the complexity of \emph{bisimilarity} for various classes of
infinite-state systems related to grammars, such as one-counter
automata, basic process algebras, and pushdown automata,
see~\cite{Burkart00} for an overview.  We mention in particular the
landmark result showing the decidability of bisimilarity for pushdown
automata~\cite{Senizergues05}.

In this paper we are concerned with \emph{probabilistic} pushdown automata
(pPDA), that is, pushdown automata with both non-deterministic and
probabilistic branching.  In particular, our pPDA generalize classical
pushdown automata without $\varepsilon$-transitions.  We refer to
automata with only probabilistic branching as \emph{fully
  probabilistic}.

We consider the complexity of checking bisimilarity for probabilistic
pushdown automata and various subclasses thereof.  The subclasses we
consider are probabilistic versions of models that have been
extensively studied in previous
works~\cite{Burkart00,SrbaVisiblyPDA:2009}.  In particular, we
consider probabilistic one-counter automata (pOCA), which are
probabilistic pushdown automata with singleton stack alphabet;
probabilistic Basic Process Algebras (pBPA), which are single-state
probabilistic pushdown automata; probabilistic visibly pushdown
automata (pvPDA), which are automata in which the stack action,
whether to push or pop, for each transition is determined by the input
letter.  Probabilistic one-counter automata have been studied in the
classical theory of stochastic processes as \emph{quasi-birth-death
  processes}~\cite{EWY10}.  Probabilistic BPA seems to have been
introduced in~\cite{Brazdil08}.

Probabilistic finite-state automata are well understood, including
the complexity of bisimilarity~\cite{Baier96,DBLP:journals/jcss/BaierEM00,CBW12}.
Probabilistic pushdown automata, 
or the equivalent model of recursive Markov chains,
have been also studied (we can name \cite{DBLP:journals/lmcs/KuceraEM06} and 
\cite{DBLP:journals/jacm/EtessamiY09} among the first respective journal papers)
but there are relatively few works on equivalence of infinite-state probabilistic
systems.  Bisimilarity of probabilistic BPA was shown decidable
in~\cite{Brazdil08}, but without any complexity bound.
In~\cite{FuKatoen11} probabilistic simulation between
pPDA and finite state systems was studied.

We also note that in this paper we consider only systems without
$\varepsilon$-transitions; in this context bisimilarity is sometimes
also called \emph{strong bisimilarity}.  In the literature there are
various notions of behavioural equivalences for systems with (silent)
$\varepsilon$-transitions, in particular \emph{weak bisimilarity} and
\emph{branching bisimilarity}. Such notions have also been studied in
the context of probabilistic systems:
see~\cite{BaierHermanns97,PhilippouLS00} for weak 
bisimilarity, see \cite{AndovaW06} for branching 
bisimilarity, and see~\cite{Zhang2018} and the references therein
for more recent variants.

Our main concern here is to extend the known algorithms for strong
bisimilarity on the above mentioned subclasses of pushdown automata to
their probabilistic extensions. In the case of weak bisimilarity, the
known results are rather negative already for non-probabilistic
systems: we can recall the undecidability result for one-counter
automata (OCA)~\cite{DBLP:conf/icalp/Mayr03}, and other results
surveyed in (the updated online version of)~\cite{Srba:roadmap:04}.

\medskip

\emph{Our contribution.} 
We first suggest a characterization of bisimilarity in a probabilistic
labelled transition system (pLTS) $\calL$ in terms of a two-player
game; the
game can be viewed as the
standard bisimulation game played in a (non-probabilistic) LTS $\calL'$
that arises from $\calL$ by adding states corresponding to probability
distributions and subsets of their supports, while the probabilities of such
subsets are ``encoded'' as new actions.
This relatively simple reduction allows us 
to leverage the rich theory that has been developed for the standard
(non-probabilistic) bisimilarity to the 
probabilistic case. This is in particular straightforward in the case 
of devices with a control unit or a stack (like pushdown automata and
their subclasses); a probabilistic machine $M$ generating a pLTS $\calL$ can be
easily transformed to a non-probabilistic $M'$ generating the
mentioned LTS $\calL'$.

A negative feature of the above transformation of $M$ to $M'$
is an exponential increase of the machine size; thus a ``blind'' use of
the reduction easily extends decidability in a standard case to the
respective probabilistic case but the upper complexity bound gets
increased. A more careful analysis of known algorithms 
working for standard $M$ 
 is required to show that these algorithms can be modified 
to be also working for probabilistic $M$ where they, in fact, decide
bisimilarity for (exponentially bigger) $M'$ without constructing $M'$
explicitly. Roughly speaking, in the standard bisimulation game the players
transform a current pair of states, i.e., of configurations of a
standard $M$,
to a new current pair in a \emph{round} of a play.
If $M$ is probabilistic, then any standard play
for $M'$ starting in a pair of configurations of
$M$ visits a pair of configurations of
$M$ every three rounds.
The mentioned modifications of standard algorithms can be viewed as 
handling such three rounds as one macro-round.

We now list concrete decidability and complexity results obtained
in this paper:
\begin{itemize}
\item Using the above-mentioned ``blind'' reduction together with the result
	of~\cite{Senizergues05} 
		(for which~\cite{JancarIcalp14} gives an alternative
		proof), we show that bisimilarity for probabilistic
  pushdown automata is decidable. We do not care about the complexity
		increase, since the problem is known to be
		non-elementary~\cite{BGKM13}
(and even Ackermann-hard~\cite{JancarFossacs14} for the model
studied in~\cite{Senizergues05}).

\item For probabilistic visibly pushdown automata (pvPDA), the reduction
  immediately yields a 2-EXPTIME upper bound,
		using the EXPTIME-completeness result 
		in~\cite{SrbaVisiblyPDA:2009} for the standard case;
the upper bound in~\cite{SrbaVisiblyPDA:2009} was achieved by a
		reduction to a result in~\cite{Walukiewicz2001}.
Here we give a self-contained 
short proof 
that the bisimilarity problem for pvPDA is in EXPTIME;  
we thus also provide a new proof of the
EXPTIME upper bound in the standard case.

\item For the class of probabilistic
  BPA, i.e., pPDA with a single control state,
		decidability 	
		was shown in~\cite{Brazdil08},
		with no complexity upper bound.
Our generic reduction
  yields a 3-EXPTIME upper bound, using the 2-EXPTIME bound in the
		standard case (stated in~\cite{Burkart00},
		and explicitly proven in~\cite{Jancar12}). 
By a detailed look at the algorithm in~\cite{Jancar12}, we show the
		modifications that place the problem in 2-EXPTIME also
		in the probabilistic case.
We note that here we have a complexity gap, 
since only an EXPTIME lower bound for this problem
is known, already in the standard case~\cite{Kiefer12}.

\item
The bisimilarity problem for 
		one-counter automata (OCA) is known to be PSPACE-complete.
We show here that the upper bound also applies to the probabilistic
		case (pOCA), by modifying the algorithm described
		in~\cite{BGJ14}.

	\item		
Finally we show that the completeness results, namely 
	the	PSPACE-completeness for pOCA (or OCA) and
	the	EXPTIME-completeness for pvPDA (or vPDA), also hold for
		\emph{fully} probabilistic OCA and vPDA, respectively.
To this aim, we adapt the respective lower bound constructions.
		(We note that~\cite{Kiefer12} 
		shows that the EXPTIME hardness also holds for fully probabilistic BPA.)
\end{itemize}
This paper is based on a conference publication~\cite{DBLP:conf/fsttcs/ForejtJKW12}. 
It has arisen by a substantial rewriting, highlights
a crucial idea in a general form, gives
complete proofs in a more unified way, and improves a 3-EXPTIME upper
bound for pBPA
from~\cite{DBLP:conf/fsttcs/ForejtJKW12} to
the 2-EXPTIME upper bound, which supports the message discussed in
Section~\ref{sec:conclusion}.

Section~\ref{sec:prelim} contains basic definitions,
and Section~\ref{sec-prob-to-nondet} then shows a game characterization 
of probabilistic bisimilarity, and a reduction to standard
bisimilarity.
Section~\ref{sec-upper-bounds}
provides the announced complexity upper bounds, while 
Section~\ref{sec-lower-bounds} shows the lower bounds.
Section~\ref{sec:conclusion} contains some concluding remarks.

\section{Preliminaries}\label{sec:prelim}

By $\N$ and $\Q$ we denote the sets of nonnegative integers and of
rationals, respectively.
Given a finite or countable set $A$, a \emph{probability distribution on} $A$ is
a function $d\colon A \rightarrow [0,1] \cap \Q$ 
such that $\sum_{a\in A}
d(a) =1$;
the \emph{support} of $d$ 
is the set $\support{d}=\{\, a\in A \mid d(a)>0\,\}$.  The set of all probability distributions on $A$ is
denoted by $\dist{A}$.
A probability distribution $d\in\dist{A}$ is {\em Dirac} if 
$d(a)=1$ for some $a\in A$ (and $d(y)=0$ for all $y\neq a$).
For any set $B \subseteq A$ we define $d(B) = \sum_{a \in B} d(a)$.
When there is no confusion, we may write $\sum_{a \in \support{d}}
d(a) a$ to indicate~$d$. E.g., for $d$ with $d(a) = \frac13$ and $d(b)
= \frac23$ we may write $\frac13 a + \frac23 b$; if $d(a)=1$, then we may
write $d$ as $1a$ or simply as $a$.

\subsection{Labelled Transition Systems (pLTSs, fpLTSs, LTSs)}

A \emph{probabilistic labelled transition system (pLTS)} is a tuple
$\calL = (S,\Sigma,\mathord{\tran{}})$, where $S$ is a finite or countable set
of \emph{states}, $\Sigma$ is a finite \emph{action alphabet}, and $\mathord{\tran{}} \subseteq S\times \Sigma \times
     {\dist{S}}$ is a \emph{transition relation}.  Throughout the
     paper we assume that pLTS are \emph{image-finite}, that is, we
     assume that for each $s\in S$ and $a\in \Sigma$ there are only
     finitely many $d\in\dist{S}$ such that
     $(s,a,d)\in\mathord{\tran{}}$.

We usually write $s\tran{a} d$ instead of
$(s,a,d)\in\mathord{\tran{}}$.  By $s \tran{} s'$ we denote that there
is a transition $s \tran{a} d$ with $s' \in \support{d}$.  (Our
definition allows that the set $\{\,s'\mid s \tran{} s'\,\}$ might be
infinite for a state $s$, though such sets are finite in the later
applications.)
State $s'$ is \emph{reachable from} $s$ if $s\tran{}^*s'$, where
$\mathord{\tran{}^*}$ is the reflexive and transitive closure of~$\mathord{\tran{}}$.

In general a pLTS combines non-deterministic and probabilistic 
 branching.  
A pLTS $\calL = (S,\Sigma,\mathord{\tran{}})$ is \emph{fully
  probabilistic}, an \emph{fpLTS}, if for 
  each pair $s\in S$, $a \in \Sigma$
we have $s \tran{a} d$ for at most one distribution~$d$.
A pLTS $\calL = (S,\Sigma,\mathord{\tran{}})$ is a \emph{standard LTS}, or just
an \emph{LTS} for short, if in each $s \tran{a} d$ the distribution
$d$ is Dirac; in
this case $s \tran{a} d$ is instead written as
 $s\trans{a}s'$ where $d(s')=1$.

Let $\calL = (S,\Sigma,\mathord{\tran{}})$ be a pLTS and $R$ be an
equivalence relation on $S$.  We say that two distributions $d,d' \in
\mathcal{D}(S)$ are \emph{$R$-equivalent} if 
$d(E) = d'(E)$ for each $R$-equivalence
class $E$.
We
furthermore say that $R$ is a \emph{bisimulation relation} if $s
\mathrel{R} t$ (and thus also  $t\mathrel{R} s$) implies that 
for each transition $s\tran{a}d$ there is a transition $t\tran{a} d'$ 
(with the same action $a$)
such
that $d$ and $d'$ are $R$-equivalent.
The union of all bisimulation
relations on $\mathcal{S}$ is itself a bisimulation relation; this
relation is called \emph{bisimulation equivalence} or
\emph{bisimilarity}~\cite{SL94}; we denote it by $\mathord{\sim}$.

We also use the following inductive characterization of bisimilarity,
assuming a pLTS $\calL = (S,\Sigma,\mathord{\tran{}})$.
We define a decreasing sequence of equivalence relations $\mathord{\sim_0}\supseteq
\mathord{\sim_1}\supseteq \mathord{\sim_2}\supseteq\cdots$ on $S$ by putting $s
\sim_0 t$ for all $s,t$, and $s\sim_{n+1}t$ if
for each transition $s\tran{a}d$ there is a transition 
$t\tran{a} d'$ such that $d,d'$ are $\mathord{\sim_n}$-equivalent
(i.e., $d(E) = d'(E)$ for 
every $\mathord{\sim_n}$-equivalence class $E$).  It is easy to verify 
that the
sequence $\mathord{\sim_n}$ converges to $\mathord{\sim}$, i.e.,
$\bigcap_{n\in\N}\mathord{\sim_n}=\mathord{\sim}$; the fact
$\bigcap_{n\in\N}\mathord{\sim_n}\supseteq \mathord{\sim}$ is trivial
(since ${\sim_n}\supseteq{\sim}$ for each $n\in\N$),
and ${\bigcap_{n\in\N}\mathord{\sim_n}}\subseteq{\sim}$ holds since
$\bigcap_{n\in\N}\mathord{\sim_n}$ is a bisimulation due to our image-finiteness
assumption on pLTSs (as can be easily checked).

\subsection{Pushdown Automata (pPDA, fpPDA, PDA) and Their Subclasses}

A \emph{probabilistic pushdown automaton} (pPDA) is a tuple $\Delta =
(Q,\Gamma,\Sigma,\mathord{\btran{}})$ where $Q$ is a finite set of
(control) \emph{states}, $\Gamma$ is a finite \emph{stack alphabet}, $\Sigma$ is a finite
\emph{action alphabet}, and $\mathord{\btran{}} \subseteq Q\times \Gamma
\times \Sigma \times {\dist{Q\times \Gamma^{\le 2}}}$
is 
a set of \emph{(transition) rules}; by  $\Gamma^{\le 2}$ we denote the
set 
$\{\varepsilon\} \cup \Gamma \cup \Gamma\Gamma$ where $\varepsilon$
denotes the empty string (and  $\Gamma\Gamma=\{\,XY\mid X\in\Gamma,
Y\in\Gamma\,\}$).
A~\emph{configuration} of $\Delta$ is a pair $(q,\beta)\in Q\times
\Gamma^*$, viewed also as a string 
$q\beta\in Q\Gamma^*$ (where $\Gamma^*$ is the set of finite words over
alphabet $\Gamma$).
We write $qX \btran{a} d$ to denote that $(q,X,a,d)$ is a
rule (i.e., an element of $\mathord{\btran{}}$).
When speaking of the \emph{size} of~$\Delta$, we assume that the
probabilities in the rules are given as
quotients of integers written in binary.

A pPDA $\Delta = (Q,\Gamma,\Sigma,\mathord{\btran{}})$ generates the
pLTS~\mbox{$\calL_{\Delta}=(Q\Gamma^*, \Sigma, \mathord{\tran{}})$}
defined as follows.  For each $\beta \in \Gamma^*$, any rule $qX \btran{a} d$ of
$\Delta$ induces a transition $qX\beta \tran{a} d'$ in
$\calL_{\Delta}$ where $d' \in \dist{Q\Gamma^*}$
satisfies $d'(p \alpha \beta) = d(p \alpha)$ for 
each $p \alpha\in\support{d}$ (hence $p\alpha\in Q\Gamma^{\leq 2}$).
We note that all configurations $q\varepsilon$ (with the empty stack) are ``dead'' (or
terminating)
states of $\calL_{\Delta}$.

A tuple $qX\in Q\Gamma$ is called a \emph{head}. 
A pPDA $\Delta = (Q,\Gamma,\Sigma,\mathord{\btran{}})$ is
\emph{fully probabilistic}, an \emph{fpPDA},
if for each head $qX$ and each action $a \in
\Sigma$ there is at most one distribution $d$ such that
 $qX \btran{a} d$.
The pLTS $\calL_{\Delta}$ generated by
an fpPDA $\Delta$ is thus an fpLTS.

A \emph{standard PDA}, or a \emph{PDA} for short, is a pPDA where all
distributions in the rules are Dirac; in this case $\calL_\Delta$ is a
(standard) LTS.

A \emph{probabilistic basic process algebra (pBPA)} is a pPDA with only one control state.
In this case it is natural to omit the control state in
configurations.

A \emph{probabilistic visibly pushdown automaton (pvPDA)} is a pPDA 
$(Q,\Gamma,\Sigma,\mathord{\btran{}})$
with a partition of the actions
 $\Sigma = \Sigmar \cup \Sigmai \cup \Sigmac$
 such that for all $pX \btran{a} d$ we have:
  if $a \in \Sigmar$ then $\support{d} \subseteq Q \times \{\varepsilon\}$;
  if $a \in \Sigmai$ then $\support{d} \subseteq Q \times \Gamma$;
  if $a \in \Sigmac$ then $\support{d} \subseteq Q \times \Gamma\Gamma$.

A \emph{probabilistic one-counter automaton (pOCA)} is a pPDA 
$\Delta=(Q,\Gamma,\Sigma,\mathord{\btran{}})$ where $\Gamma=\{I,Z\}$,
for each $p\alpha\in Q\Gamma^{\leq 2}$ such that 
$d(p\alpha)>0$  for some $qI\btran{a}d$ we have $\alpha\in
\{\varepsilon, I, II\}$, and for each 
$p\alpha\in Q\Gamma^{\leq 2}$ such that 
$d(p\alpha)>0$  for some $qZ\btran{a}d$ we have $\alpha\in
\{Z,IZ\}$. 
In this case $\calL_\Delta$ is restricted to the set of states $q\alpha Z$
where $\alpha\in\{I\}^*$; this set is closed under reachability due to
the form of $\mathord{\btran{}}$. In  $q\alpha Z$ the length of $\alpha$ is viewed
as a non-negative counter value, while $Z$ always and only occurs at
the bottom of the stack. 

The fully probabilistic and standard versions of the mentioned
subclasses of pPDA, i.e., \emph{fpBPA},  \emph{fpvPDA}, \emph{fpOCA},
\emph{BPA}, \emph{vPDA}, \emph{OCA},
are defined analogously as in the general case.

The \emph{bisimilarity problem} for pPDA asks whether two configurations
$q_1 \alpha_1$ and $q_2 \alpha_2$ of a given pPDA $\Delta$ are bisimilar when
regarded as states of the generated pLTS $\calL_{\Delta}$.
We will also consider restrictions of the problem to subclasses of
pPDA.

\begin{figure}
\centering
\begin{tikzpicture}[x=2.5cm,y=1.7cm] 
\foreach \x/\l in {0/pZ,1/pXZ,2/pXXZ,3/pXXXZ,4/pXXXXZ}
   \node (D\x)   at (\x,-1)   [ran] {$\l$};
\foreach \x/\xx in {1/0,2/1,3/2,4/3}
   \draw [tran] (D\x) to  node[below]   {$0.5$} (D\xx);
\foreach \x/\l in {2/qXXZ,3/qXXXZ,4/qXXXXZ}
   \node (H\x)   at (\x,0)   [ran] {$\l$};
\foreach \x/\xx in {1/2,2/3,3/4}
   \draw [tran,rounded corners] (D\x) -- node[very near start,above=0.7mm] {$0.5$} (H\xx);
\foreach \x/\xx in {2/3,3/4}
   \draw [tran,rounded corners] (H\x) -- node[near start,above=0.5mm,inner sep=0] {$1$} (D\xx);
\draw [thick, dotted] (H4) -- +(.8,0);
\draw [thick, dotted] (D4) -- +(.8,0);

\foreach \x/\l in {0/r,1/rX,2/rYX,3/rXXX}
   \node (A\x)   at (\x,-2)   [ran] {$\l$};
   \node (Ba)   at (2,-2.5)   [ran] {$rYX'$};
   \node (Bb)   at (3,-3)   [ran] {$rXXX'$};

   \node (Aa)   at (4,-1.9)   [ran] {$rXX$};
   \node (Ab)   at (4,-2.3)   [ran] {$rYXXX$};
   \node (Ac)   at (4,-2.7)   [ran] {$rYX'XX$};

   \draw [tran,rounded corners] (A1) -- node[midway,below] {$0.5$} (A0);
   \draw [tran,rounded corners] (A1) -- node[midway,above] {$0.3$} (A2);
   \draw [tran,rounded corners] (A1) -- node[midway,below] {$0.2$} (Ba);
   \draw [tran,rounded corners] (A2) -- node[midway,above] {$1$} (A3);
   \draw [tran,rounded corners] (A3) -- node[pos=0.7,above] {$0.5$} (Aa.west);
   \draw [tran,rounded corners] (A3) -- node[near end,above] {$0.3$} (Ab.west);
   \draw [tran,rounded corners] (A3) -- node[pos=0.5,below=1mm] {$0.2$} (Ac.west);
   \draw [tran,rounded corners] (Ba) -- node[midway,above] {$1$} (Bb);
\draw [thick, dotted] (Aa) -- +(.8,0);
\draw [thick, dotted] (Ab) -- +(.8,0);
\draw [thick, dotted] (Ac) -- +(.8,0);

\draw [thick, dotted] (Bb) -- +(.8,-0.25);

\end{tikzpicture}
\caption{A fragment of $\calL_{\Delta}$ from Example~\ref{ex-preli-example}.}
\label{fig-preli-example}
\end{figure}
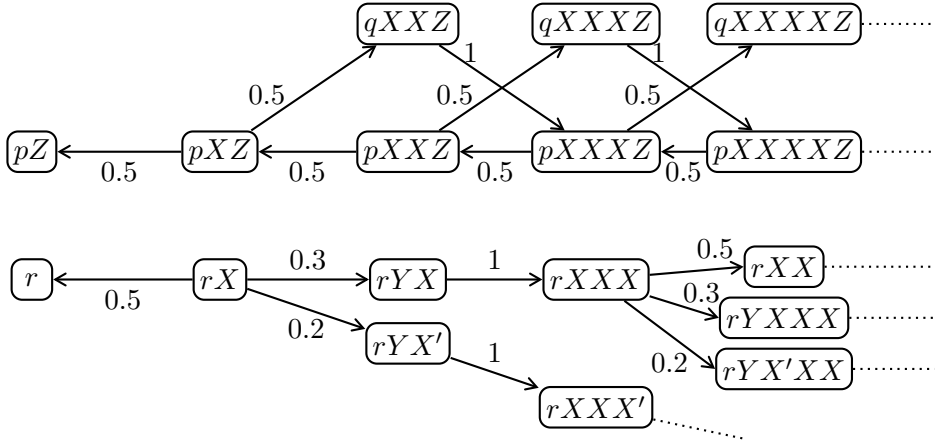

\begin{exa} \label{ex-preli-example}
  Consider the fpPDA~$\Delta = (\{p,q,r\},\{X,X',Y,Z\},\{a\},\mathord{\btran{}})$ with the following rules
   (omitting the unique action~$a$):
  \begin{align*}
     pX &\btran{} 0.5 qXX + 0.5 p               & qX &\btran{} pXX \\
     rX &\btran{} 0.3 rYX  + 0.2  rYX' + 0.5 r  & rY &\btran{} rXX \\
     rX' &\btran{} 0.4 rYX + 0.1 rYX' + 0.5 r
  \end{align*}
  The restriction of~$\Delta$ to the control states $p,q$ and to the
	stack symbols $X,Z$ yields a pOCA, when $X$ plays the role of
	$I$ from the definition.
  The restriction of~$\Delta$ to the control state~$r$ and the stack symbols $X,X',Y$ yields a pBPA.
	A fragment of the pLTS $\calL_{\Delta}$ is shown in Figure~\ref{fig-preli-example}.
  The configurations $pXZ$ and~$rX$ are bisimilar,
   as there is a bisimulation relation with equivalence classes
    $\{pX^kZ\} \cup \{\,r w \mid w \in \{X,X'\}^k\,\}$ for all $k \ge 0$ and
    $\{qX^{k+1}Z\} \cup \{\,r Y w \mid w \in \{X,X'\}^k\,\}$ for all $k \ge 1$.
\end{exa}

\section{Game Characterization of Probabilistic Bisimilarity}
\label{sec-prob-to-nondet}

Bisimilarity in standard LTSs has a natural characterization in terms
of two-player games.  This game-theoretic characterization is a source
of intuition and allows for elegant presentations of certain
arguments.  The players in the bisimulation game can be called
\emph{Attacker} and \emph{Defender}; a play of such a game, in an LTS
$\calL=(S,\Sigma,\mathord{\tran{}})$, yields a sequence of pairs of states,
starting from a given initial pair $(s,s')$.  The objective of the
game for \emph{Attacker} is a reachability objective.  Specifically,
pairs of states $(s_1,s_2)$ for which $s_1$ and $s_2$ enable different
sets of actions are declared as Attacker goals; Attacker aims to reach
such a pair, Defender aims to avoid this.

A play arises as follows. 
Given a current pair $(s_1,s_2)$, the players  perform
a prescribed protocol that results either in a demonstration 
that  $(s_1,s_2)$ is an Attacker goal, in which case the play finishes,
or in creating a new current pair 
 $(s'_1,s'_2)$. 
In the standard bisimilarity game the protocol is the
 following:
 Attacker chooses a transition $s_i\tran{a}s'_i$ for some
 $i\in\{1,2\}$, $a\in\Sigma$, $s'_i\in S$,
 and Defender responds by choosing
 $s_{3-i}\tran{a}s'_{3-i}$ for some $s'_{3-i}\in S$; this yields a pair $(s'_1,s'_2)$.
 If Defender has no possible response (action $a$ is not enabled in
 $s_{3-i}$), then the play finishes since 
 it has been demonstrated that $(s_1,s_2)$ is an
 Attacker goal. If no action is enabled in any of $s_1,s_2$,
 then we can formally put  $(s'_1,s'_2)=(s_1,s_2)$, or simply stop the
 play as a win of Defender.
 
 It is easy to observe that $s\sim_n s'$ 
iff Attacker cannot force
 his win in $n$ rounds of the play (i.e., within $n$ runnings of the
 protocol) when $(s,s')$ is the initial current pair. 

In the case of \emph{probabilistic bisimilarity}, i.e.,
of bisimilarity in probabilistic LTSs,
a similar game characterization is not immediately obvious.
We suggest the characterization arising by 
the following modification of the above game, now in a pLTS
$\calL=(S,\Sigma,\mathord{\tran{}})$; we define
the protocol, to be performed in a round of the game 
from a current pair
$(s_1,s_2)$, as follows:

\begin{enumerate}
\item
Attacker chooses $s_i\tran{a}d_i$, for some
	$i\in\{1,2\}$ and $(s_i,a,d_i)\in\mathord{\tran{}}$; if this is not possible, the play stops by
		declaring Defender the winner.
		Defender chooses $s_{3-i}\tran{a}d_{3-i}$
		for some  $(s_{3-i},a,d_{3-i})\in\mathord{\tran{}}$; if not
		possible, then the play finishes
		by declaring Attacker the winner 
		(it has been demonstrated that
		an Attacker goal has been reached).
\item
Attacker chooses a nonempty subset $T_j\subseteq \support{d_j}$
for some
		$j\in\{1,2\}$; Defender chooses some $T_{3-j}\subseteq
		\support{d_{3-j}}$ such
that $d_{3-j}(T_{3-j})\geq d_{j}(T_{j})$.
\item
Attacker chooses $s'_k\in T_k$ for some
		$k\in\{1,2\}$; 	Defender chooses $s'_{3-k}\in T_{3-k}$.
	\item
	The pair $(s'_1,s'_2)$ becomes a new current pair.
\end{enumerate}

\begin{exa}\label{ex-gen-nondet-constr1}
Figure~\ref{fig-ex-gen-nondet-constr} shows a finite pLTS $\calL = (S,\Sigma,\mathord{\tran{}})$.
\begin{figure}
\begin{center}
\begin{tikzpicture}[xscale=2,thick]
\tikzstyle{lran} = [ran,inner sep=2pt,minimum height=5mm, minimum width=8mm];

\node[lran] (s1) at (0,0)  {$s$};
\node[lran] (t1) at (2,1)  {$t_1$};
\node[lran] (t2) at (3,1)  {$t_2$};
\coordinate (t1t) at (2.5,2.5);
\coordinate (t1b) at (2.5,-0.5);
\node[lran] (u1) at (2.5,-1.5) {$u$};
\coordinate (s1t1) at (1.2,0);

\draw[tran] (t1) edge [bend left] node[above] {$a$} (t2);
\draw[tran] (t2) edge [bend left] node[below] {$a$} (t1);
\draw[tran, rounded corners=4mm] (s1) to [bend left] node[above] {$a$} (t1t) to node[left] {$\frac13$} (t1);
\draw[tran, rounded corners=4mm] (s1) to [bend left]                   (t1t) to node[right] {$\frac23$} (t2);
\draw[tran, rounded corners=2mm] (u1) to node[pos=0.5,right,inner sep=0.5mm] {$b$} (t1b) to node[left] {$\frac13$} (t1);
\draw[tran, rounded corners=2mm] (u1) to                                          (t1b) to node[right] {$\frac23$} (t2);
\draw[tran, rounded corners=2mm] (s1) to                   +(1,0) to node[above] {$\frac12$} (t1);
\draw[tran, rounded corners=2mm] (s1) to node[above] {$b$} +(1,0) to node[below] {$\frac12$} (u1);
\draw[tran] (u1) to [bend left=15] node[left] {$a$} (t1);
\draw[tran] (t2) edge [loop,out=20,in=-20,looseness=5] node[right] {$a$} (t2);
\end{tikzpicture}
\end{center}
\caption{A finite pLTS $\calL$ with bisimulation equivalence classes 
$\{s\}, \{t_1, t_2\}, \{u\}$.}\label{fig-ex-gen-nondet-constr}
\end{figure}
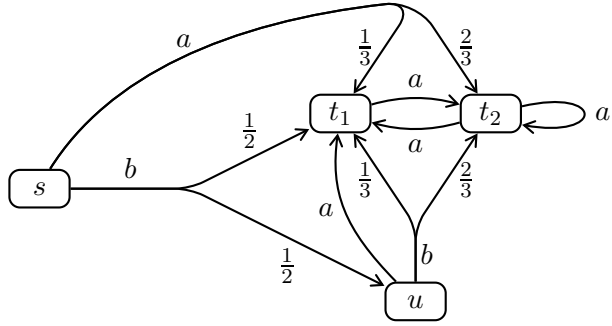
Here is a play that proves $s \not\sim_2 u$.
The play starts in $(s, u)$.
Attacker chooses $s \tran{b} d_1$ where $d_1 = \frac12 u + \frac12 t_1$.
This forces Defender to choose $u \tran{b} d_2$ where $d_2 = \frac13 t_1 + \frac23 t_2$.
Now Attacker may choose $\{t_2\} \subseteq \support{d_2}$, where $d_2(\{t_2\}) = \frac23$.
Defender now has to choose $T_1 \subseteq \support{d_1} = \{u,t_1\}$ so that $d_1(T_1) \ge \frac23$, i.e., she has to choose $T_1 = \{u, t_1\}$.
Attacker can now take $u \in T_1$, and Defender can only take $t_2$.
Thus the new pair is $(u,t_2)$.
In the next round Attacker chooses $u \tran{b} d_2$, and Defender cannot respond.
That is, $(u,t_2)$ is an Attacker goal, and Attacker has won.
Thus $s \not\sim_2 u$.
\end{exa}

It might not be obvious in general that $s\sim_n s'$ (in the given pLTS) 
iff Attacker cannot force a win in $n$ rounds when starting from
$(s,s')$.
We prove this formally below; moreover, we intentionally implement the protocol
as a standard bisimulation (sub)game (consisting of three rounds),
since this makes it easier
to lift known results for standard bisimilarity to the case of
probabilistic bisimilarity.
Hence we transform a given probabilistic LTS $\calL$ to a
standard LTS $\calL'$ that extends the state set of $\calL$ 
by viewing distributions and also 
nonempty subsets of their supports as explicit states; we will
guarantee that a pair of
original states will be bisimilar in (the probabilistic) $\calL$ iff
it is bisimilar in (the standard) $\calL'$.

Now we describe the transformation of $\calL$ to $\calL'$ rigorously.
We assume a pLTS $\calL = (S,\Sigma,\mathord{\tran{}})$ and use the following
technical notions:
\begin{itemize}
\item
a \emph{distribution}
$d\in\dist{S}$ is \emph{relevant} if
$s\tran{a}d$ for some $s\in S$, $a\in\Sigma$;
\item
 a \emph{nonempty set of states}
$T\subseteq S$ is
\emph{relevant} if $T\subseteq\support{d}$ for some relevant $d$;
\item
	a \emph{number} $\rho$ is \emph{relevant} if 
	$\rho=d(T)$ for some relevant $d$ and  $\emptyset\neq T\subseteq \support{d}$
	
	(hence $0<\rho\leq 1$).
\end{itemize}

The LTS $\calL' = (S',\Sigma',\mathord{\ctran{}})$ is defined as follows:
\begin{itemize}
\item
$S'=S\cup\{\,d\in\dist{S}\mid d$ is relevant$\,\}\cup\{\,T\subseteq S\mid
T$ is relevant$\,\}$.
Note that $S'$ is at most countable if each relevant distribution~$d$ has finite support.
\item
$\Sigma'=\Sigma\cup\{\,\rho\in[0,1]\mid \rho$ is relevant$\,\}\cup\{\#\}$ where
the three parts of the union are pairwise disjoint.
\item
The relation $\mathord{\ctran{}}$ is the smallest relation satisfying the
		following conditions:
\begin{itemize}
\item
if $s\tran{a}d$ (in $\calL$), then 
$s\ctran{a}d$ (in $\calL'$);
\item
if $T\subseteq\support{d}$ and $d(T)\geq \rho$ for some relevant
		$d,\rho$ (hence $T\neq\emptyset$),
then $d\ctran{\rho}T$;
\item
if $s\in T$ (for some relevant $T$), then 
$T\ctran{\#}s$.
\end{itemize}
\end{itemize}

\begin{exa}\label{ex-gen-nondet-constr2}
Figure~\ref{fig-ex-gen-nondet-constr2} shows the LTS $\calL' =
(S',\Sigma',\mathord{\ctran{}})$ obtained from~$\calL$ in
Example~\ref{ex-gen-nondet-constr1} (depicted in Figure~\ref{fig-ex-gen-nondet-constr}) according to the construction above.
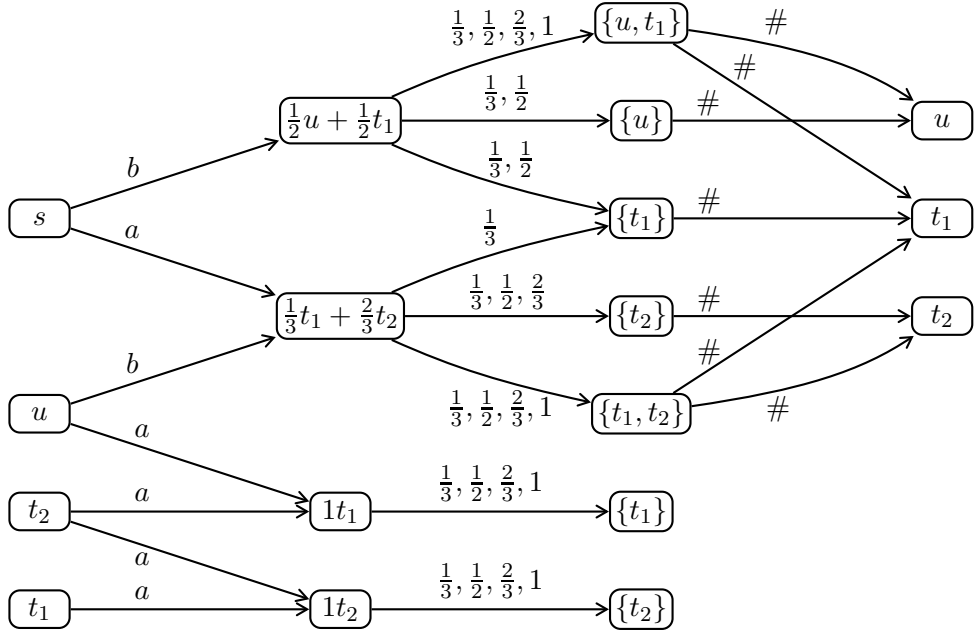
\begin{figure}
\begin{center}
\begin{tikzpicture}[xscale=4,yscale=1.3,thick]
\tikzstyle{lran} = [ran,inner sep=2pt,minimum height=5mm, minimum width=8mm];

\node[lran] (u)    at (1,-1) {$u$};
\node[lran] (t1)   at (1,-2) {$t_1$};
\node[lran] (t2)   at (1,-3) {$t_2$};

\node[lran] (Sut1) at (0,0)  {$\{u, t_1\}$};
\node[lran] (Su)   at (0,-1) {$\{u\}$};
\node[lran] (St1)  at (0,-2) {$\{t_1\}$};
\node[lran] (St2)  at (0,-3) {$\{t_2\}$};
\node[lran] (St1t2)at (0,-4) {$\{t_1,t_2\}$};
\node[lran] (St1') at (0,-5) {$\{t_1\}$};
\node[lran] (St2') at (0,-6) {$\{t_2\}$};

\node[lran] (dt1u) at (-1,-1) {$\frac12 u + \frac12 t_1$};
\node[lran] (dt1t2)at (-1,-3) {$\frac13 t_1 + \frac23 t_2$};
\node[lran] (dt1)  at (-1,-5) {$1 t_1$};
\node[lran] (dt2)  at (-1,-6) {$1 t_2$};

\node[lran] (s)    at (-2,-2) {$s$};
\node[lran] (u')   at (-2,-4) {$u$};
\node[lran] (t2')  at (-2,-5) {$t_2$};
\node[lran] (t1')  at (-2,-6) {$t_1$};

\draw[tran] (Sut1) to[bend left=20] node[pos=0.3 ,above          ] {$\#$}                           (u); 
\draw[tran] (Sut1) to               node[pos=0.3 ,above          ] {$\#$}                           (t1); 
\draw[tran] (Su)   to               node[pos=0.15,above,inner sep=0.5mm]  {$\#$}                    (u); 
\draw[tran] (St1)  to               node[pos=0.15,above,inner sep=0.5mm]  {$\#$}                    (t1); 
\draw[tran] (St2)  to               node[pos=0.15,above,inner sep=0.5mm]  {$\#$}                    (t2); 
\draw[tran] (St1t2)to               node[pos=0.15,above,inner sep=0.5mm]  {$\#$}                    (t1); 
\draw[tran] (St1t2)to[bend right=20]node[pos=0.30,below,inner sep=0.5mm]  {$\#$}                    (t2); 

\draw[tran] (dt1u) to[bend left=10] node[pos=0.6,above] {$\frac13, \frac12, \frac23, 1$} (Sut1);
\draw[tran] (dt1u) to               node[        above          ] {$\frac13, \frac12            $} (Su);
\draw[tran] (dt1u) to[bend right=10]node[pos=0.6,above]{$\frac13, \frac12            $} (St1);

\draw[tran] (dt1t2)to[bend left=10] node[pos=0.5,above] {$\frac13                     $} (St1);
\draw[tran] (dt1t2)to               node[        above,yshift=-0.5mm ] {$\frac13, \frac12, \frac23   $} (St2);
\draw[tran] (dt1t2)to[bend right=10]node[pos=0.6,below]{$\frac13, \frac12, \frac23, 1$} (St1t2);

\draw[tran] (dt1)  to               node[        above          ] {$\frac13, \frac12, \frac23,1 $} (St1');
\draw[tran] (dt2)  to               node[        above          ] {$\frac13, \frac12, \frac23,1 $} (St2');

\draw[tran] (s)    to               node[pos=0.3,above          ] {$b$                           } (dt1u);
\draw[tran] (s)    to               node[pos=0.3,above          ] {$a$                           } (dt1t2);
\draw[tran] (u')   to               node[pos=0.3,above          ] {$b$                           } (dt1t2);
\draw[tran] (u')   to               node[pos=0.3,above          ] {$a$                           } (dt1);
\draw[tran] (t2')  to               node[pos=0.3,above          ] {$a$                           } (dt1);
\draw[tran] (t2')  to               node[pos=0.3,below,inner sep=0.9mm] {$a$                           } (dt2);
\draw[tran] (t1')  to               node[pos=0.3,above          ] {$a$                           } (dt2);
\end{tikzpicture}
\end{center}
\caption{The standard LTS $\calL'$ obtained from the probabilistic LTS
	$\calL$ in Figure~\ref{fig-ex-gen-nondet-constr}.
For better readability, some states appear multiple times but their outgoing transitions are drawn only once.}
\label{fig-ex-gen-nondet-constr2}
\end{figure}
The play from Example~\ref{ex-gen-nondet-constr1} has a corresponding play in the bisimulation game of~$\calL'$.
We have $s \sim_1 u$ in~$\calL$ and $s \sim_3 u$ in~$\calL'$, but $s \not\sim_2 u$ in~$\calL$ and $s \not\sim_4 u$ (hence $s \not\sim_6 u$) in~$\calL'$.
\end{exa}

The next lemma captures a crucial fact on the relation of $\calL$ and
$\calL'$.

\begin{lem}\label{lem:gen-nondet-constr-correct}
Given a pLTS $\calL$ and the LTS $\calL'$ as above,
for any states $s,s'$ of $\calL$ and any $n\in\N$ we have
		\begin{center}
	$s\sim_n s'$ in $\calL$ iff $s \sim_{3n} s'$ in $\calL'$.
	\end{center}
This also yields that
$s\sim s'$ in $\calL$ iff $s \sim s'$ in $\calL'$.
	\end{lem}

\begin{proof}
We assume $\calL = (S,\Sigma,\mathord{\tran{}})$ and $\calL' =
	(S',\Sigma',\mathord{\ctran{}})$ as above.
	We recall that $d_1,d_2\in\dist{S}$ are called $R$-equivalent, for an equivalence $R$ on
	$S$, iff $d_1(E)=d_2(E)$ 
for each $R$-class $E$.
We use the expression $d_1\sim_n d_2$ to denote
that $d_1,d_2$ are	$\mathord{\sim_n}$-equivalent (for $\mathord{\sim_n}$ in $\calL$).

We also say that subsets $T_1,T_2$ of $S$ are \emph{$R$-similar} if
	$T_1,T_2$ represent the same $R$-classes, i.e.,
if the sets $\{\,E\mid E$ is an $R$-class, $E\cap T_1\neq\emptyset\,\}$
and $\{\,E\mid E$ is an $R$-class, $E\cap T_2\neq\emptyset\,\}$ are the
	same.
By $T_1\sim_n T_2$ we denote that  $T_1,T_2$ are
$\mathord{\sim_n}$-similar  (for $\mathord{\sim_n}$ in $\calL$).

	By $\mathord{\sim}$ we denote bisimilarity in (the pLTS) $\calL$, and by 
 $\mathord{\sim'}$ we denote bisimilarity in
	(the standard LTS)
	$\calL'$.
It suffices to show that for each $n\in\N$ we have:	
	\begin{enumerate}
	\item
$s_1\sim_n s_2$ iff $s_1 \sim'_{3n} s_2$.
	\item
$d_1\sim_n d_2$ iff $d_1 \sim'_{3n+2} d_2$.
\item 
$T_1\sim_n T_2$ iff $T_1 \sim'_{3n+1} T_2$.
	\end{enumerate}
	(By $s_i$ we denote elements of $S$, by $d_i$ relevant
	distributions, and by $T_i$ relevant sets.)

We show Claims 1--3 by the following observations:
\begin{itemize}
\item 
Claim $1$ trivially holds for $n=0$.
\item
If Claim $1$ holds for $n$, then Claim $3$ holds for $n$:

			We observe that $T_1\sim_{n}T_2$ iff 
			for each $i\in\{1,2\}$ and each transition 
			$T_i\ctran{\#}s_i$ (in $\calL'$)
there is a transition $T_{3-i}\ctran{\#}s_{3-i}$ such that 
			$s_1\sim_n s_2$.
			Assuming that $s_1\sim_n s_2$ iff
			$s_1\sim'_{3n} s_2$ (Claim 1), we thus get 
that $T_1\sim_{n}T_2$	
iff  $T_1\sim'_{3n+1}T_2$ (Claim 3).

\item
If Claim $3$ holds for $n$, then Claim $2$ holds for $n$:

Let $d_1 \sim_{n} d_2$, $i\in\{1,2\}$, and  
$d_i\ctran{\rho}T_i$ (hence $d_i(T_i)\geq
			\rho$). 
			Put $T_{3-i}=\{\,s'\in\support{d_{3-i}}\mid
			 s'\sim_n s$ for some $s\in T_i\,\}$. 
 For each $\mathord{\sim_n}$-class $E$ we have $d_1(E)=d_2(E)$, and
			 thus $d_i(T_i\cap E)\leq d_{3-i}(T_{3-i}\cap E)$.
			 Hence $d_i(T_i)\leq d_{3-i}(T_{3-i})$, and thus
			$d_{3-i}\ctran{\rho}T_{3-i}$; moreover,
			in 	$T_1,T_2$
			the same $\mathord{\sim_n}$-classes are represented,
			i.e., 
			$T_1\sim_n T_2$, hence $T_1\sim'_{3n+1} T_2$
			(assuming Claim 3).
		Therefore $d_1 \sim'_{1+(3n+1)} d_2$, i.e., 
		$d_1 \sim'_{3n+2} d_2$.

Now let $d_1 \not\sim_{n} d_2$; there is thus a $\mathord{\sim_n}$-class $E$
			such that $d_i(E)>d_{3-i}(E)$ for some
			$i\in\{1,2\}$. 
			For the transition $d_i\ctran{\rho}T_i$ where
			$T_i=\support{d_i}\cap E$ and $\rho=d_i(T_i)$
and any transition
			$d_{3-i}\ctran{\rho}T_{3-i}$ we have
			$T_1\not\sim_n T_2$; indeed, from $d_{3-i}(T_{3-i})\geq
			\rho$
			we deduce
			$d_{3-i}(T_{3-i})\geq \rho=d_i(T_i)=
			d_i(E)>d_{3-i}(E)$, and the fact 
			that  $T_{3-i}\subseteq\support{d_{3-i}}$ then
entails the existence of some $s\in T_{3-i}\smallsetminus E$.
Since $T_1\not\sim_n T_2$ entails
			$T_1\not\sim'_{3n+1}T_2$ (assuming Claim 3), we get 
		$d_1 \not\sim'_{1+(3n+1)} d_2$, i.e., $d_1 \not\sim'_{3n+2} d_2$.
\item
	If Claim $2$ holds for $n$, then Claim $1$ holds for $n{+}1$:

By definition (of relations $\mathord{\sim_n}$ in a pLTS), $s_1\sim_{n+1} s_2$
iff for each $i\in\{1,2\}$ and each transition 
			$s_i\trans{a}d_i$ (in $\calL$) there is 
			a transition $s_{3-i}\trans{a}d_{3-i}$ such
			that $d_1\sim_{n}d_2$. 
			By definition of $\calL'$ 
			the only transitions from a state $s\in S$ in $\calL'$
			are the transitions $s\ctran{a}d$ 
			where $s\trans{a}d$ in $\calL$.
			With the assumption that $d_1\sim_{n}d_2$ iff
			$d_1\sim'_{3n+2}d_2$ (Claim 2) we
thus easily verify that 
$s_1\sim_{n+1} s_2$ iff  $s_1\sim'_{1+(3n+2)} s_2$, i.e., 
$s_1\sim_{n+1} s_2$ iff  $s_1\sim'_{3(n+1)} s_2$. \qedhere
\end{itemize}
\end{proof}
\noindent
\emph{Remark.} 
We have confined ourselves to pLTSs $\calL=(S,\Sigma,\mathord{\tran{}})$ that
are image-finite and where $\Sigma$ is finite and $S$ is at most
countable. This framework is sufficient for our intended applications,
but the above transformation of a pLTS $\calL$ to the LTS $\calL'$ can
be applied to general pLTSs as well (where 
also equivalences $\mathord{\sim_\lambda}$ for infinite ordinals $\lambda$ are
considered).

\section{Upper Bounds for Subclasses of pPDA}\label{sec-upper-bounds}

The transformation of a pLTS $\calL$ to an LTS $\calL'$ described in
Section~\ref{sec-prob-to-nondet} is now applied to derive decidability
and complexity results for bisimilarity for pPDA (and subclasses
thereof) from known results for standard (i.e., non-probabilistic)
versions.

We use the fact that if $\calL$ is generated by a probabilistic model $M$
with a finite control unit
or a stack (which is, in particular, the case of pPDA, pBPA, pOCA,
pvPDA), 
then there are straightforward transformations of $M$ to
non-probabilistic versions $M'$ that represent $\calL'$.  An explicit
presentation of $M'$ might be exponentially larger than $M$, so using
complexity results for standard models as a ``black box'' to derive
complexity results for probabilistic models incurs an exponential
complexity blow-up.  This complexity increase is not significant for
general pPDA, since bisimilarity for PDA is known to be
non-elementary~\cite{BGKM13} (it is even
Ackermann-hard~\cite{JancarFossacs14} for the model studied
in~\cite{Senizergues05}).  On the other hand, in the cases of standard
OCA, vPDA, BPA, the known upper bounds are PSPACE, EXPTIME, and
2-EXPTIME, respectively. (For OCA and vPDA the upper bounds match the
known lower bounds, the problem for BPA is only known to be
EXPTIME-hard.)  

Since $M'$ arises from $M$ by a specific enhancement causing only a
``local'' exponential increase that can be left implicit (i.e., $M$
can be viewed as representing $M'$ without an explicit construction of
$M'$), by studying the algorithms for standard cases we might be able
to verify that the exponential increase can be avoided.  We will show
that this is indeed the case for pvPDA, pBPA, and pOCA: here we will
argue that the exponential increase only plays a role in the changed
protocol (that transforms a current pair of states to a new one),
which can be still performed in polynomial time with bounded
alternation and does not affect the mentioned PSPACE, EXPTIME,
2-EXPTIME upper bounds.  In the general case of pPDA we use the
``black-box'' reduction (in Section~\ref{sec-pPDA-decidable}), since a
similar argument that the complexity bound does not increase would
require recalling the involved algorithms in the standard
case~\cite{Senizergues05,JancarIcalp14} and, as we have mentioned above,
the complexity is already non-elementary in the standard case.

\subsection{Bisimilarity of pPDA is
Decidable}\label{sec-pPDA-decidable}

As announced, here we show a ``black-box'' use 
of the above transformation in the general
case of pPDA.
We have introduced PDA as pPDA with only Dirac distributions.  In
fact, PDA are the standard nondeterministic pushdown automata (with no
$\varepsilon$-transitions).  For PDA, bisimilarity is known to be
decidable~\cite{Senizergues05} (see also~\cite{JancarIcalp14}), with
no explicit complexity upper bound.  To show the decidability of
bisimilarity for pPDA, it suffices to show how to realize a
transformation of a pPDA $\Delta$, representing a pLTS $\calL$, to a
PDA $\Delta'$ representing the LTS $\calL'$ as defined in
Section~\ref{sec-prob-to-nondet}.

This task is straightforward, since we can naturally represent the
relevant distributions and the subsets of their supports by using
either the control state unit or (the top of) the stack. 
Because of the mentioned subsets of supports, we can thus increase the size of the
control unit or of the stack alphabet exponentially; the size of 
the
action alphabet increases similarly (due to actions $\rho\in\Q$).
We now describe the ``stack option'' in detail.

\medskip

For a pPDA $\Delta = (Q,\Gamma,\Sigma,\mathord{\btran{}})$
we define the PDA $\Delta' =
(Q,\Gamma',\Sigma',\mathord{\ctran{}})$ 
as follows:
\begin{itemize}
 \item The stack alphabet $\Gamma'$ arises from 
   $\Gamma$ by adding the following fresh symbols:
for every distribution $d$ such that $\Delta$ contains 
a rule $pX \btran{a} d$
 we add a symbol $\stacksymba{d}$; for every nonempty
		set $T$ such that 
		$T\subseteq \support{d}$ for some of the above
		distributions $d$ 		
		we add a symbol $\stacksymba{T}$.

 \item
 $\Sigma'=\Sigma\cup W\cup\{\#\}$
where the parts of the union are pairwise disjoint
		and $W=\{\,\rho\in \mathbb{Q}\mid \rho=d(T)$ for some
		rule $pX \btran{a} d$ and $\emptyset\neq
T\subseteq\support{d}\,\}$.
 \item
The set $\mathord{\ctran{}}$ of rules of $\Delta'$
		is defined as follows.
We choose an arbitrary state $q_0\in Q$.		
Every
rule $qX \btran{a} d$ of $\Delta$ gives rise to the following rules of
$\Delta'$:
		\begin{itemize}		
			\item
			$qX\ctran{a}q_0\stacksymba{d}$;
		\item
		$q_0\stacksymba{d} \ctran{\rho} q_0\stacksymba{T}$ for
				all $\rho \in W$ and 
		$T\subseteq \support{d}$ where $d(T)\geq \rho$;
	\item				
$q_0\stacksymba{T} \ctran{\#} p\alpha$ for 				
				each above defined
symbol $\stacksymba{T}$ and $p\alpha \in T$.
		\end{itemize}
		\end{itemize}

\begin{exa}\label{ex-gen-constr-PDA}
Consider the pPDA $\Delta = (\{p,q\}, \{X, Y\}, \{a,b\}, \mathord{\btran{}})$ with the following rules:
\[ \textstyle
  p X \btran{a} \frac13 q + \frac23 p Y X, \qquad 
  q Y \btran{a} \frac13 p + \frac23 p X, \qquad
  p Y \btran{b} q Y, \qquad
  q Y \btran{a} p X Y.
\]
The construction above yields the PDA $\Delta' = (\{p,q\}, \Gamma', \Sigma', \mathord{\ctran{}})$ where
\begin{align*} 
	\Gamma' \ &= \  \big\{X,Y\big\}\cup
	\big\{\stacksymba{\frac13 q + \frac23 p
Y X}, \stacksymba{\frac13 p + \frac23 p X}, \stacksymba{1 q Y},
\stacksymba{1 p X Y}\big\}\,\cup \\
& \hspace{2em} \cup\big\{\stacksymba{\{q, p Y X\}},
\stacksymba{\{q\}}, \stacksymba{\{p Y X\}}, \stacksymba{\{p, p X\}},
\stacksymba{\{p\}}, \stacksymba{\{p X\}}, \stacksymba{\{q Y\}},
\stacksymba{\{p X Y\}}\big\}, \\
\Sigma' \ &= \ \textstyle\{a, b, \frac13, \frac23, 1, \#\},
\end{align*}
and, choosing $q_0 = p$, the rules in $\mathord{\ctran{}}$ are
given in Figure~\ref{fig:rulesofexamplepda}.
\begin{figure}[hbt]
\begin{gather*} 
p X \ctran{a} p 
\stacksymba{\frac13 q + \frac23 p Y X} \qquad 
q Y \ctran{a} p \stacksymba{\frac13 p + \frac23 p X} \qquad 
p Y \ctran{b} p \stacksymba{1 q Y} \qquad 
q Y \ctran{a} p \stacksymba{1 p X Y}
\\
\begin{aligned}
 p \stacksymba{\frac13 q + \frac23 p Y X} &\ctran{1} p \stacksymba{\{q, p Y X\}} 
  \qquad& p \stacksymba{\frac13 p + \frac23 p X} &\ctran{1} p \stacksymba{\{p, p X\}} \\
 p \stacksymba{\frac13 q + \frac23 p Y X} &\ctran{\frac23} p \stacksymba{\{q, p Y X\}}
  & p \stacksymba{\frac13 p + \frac23 p X} &\ctran{\frac23} p \stacksymba{\{p, p X\}} \\
 p \stacksymba{\frac13 q + \frac23 p Y X} &\ctran{\frac13} p \stacksymba{\{q, p Y X\}} 
  & p \stacksymba{\frac13 p + \frac23 p X} &\ctran{\frac13} p \stacksymba{\{p, p X\}} \\
 p \stacksymba{\frac13 q + \frac23 p Y X} &\ctran{\frac13} p \stacksymba{\{q\}} 
  & p \stacksymba{\frac13 p + \frac23 p X} &\ctran{\frac13} p \stacksymba{\{p\}} \\
 p \stacksymba{\frac13 q + \frac23 p Y X} &\ctran{\frac23} p \stacksymba{\{p Y X\}} 
  & p \stacksymba{\frac13 p + \frac23 p X} &\ctran{\frac23} p \stacksymba{\{p X\}} \\
 p \stacksymba{\frac13 q + \frac23 p Y X} &\ctran{\frac13} p \stacksymba{\{p Y X\}}
  & p \stacksymba{\frac13 p + \frac23 p X} &\ctran{\frac13} p \stacksymba{\{p X\}} \\
 p \stacksymba{1 q Y} &\ctran{1} p \stacksymba{\{q Y\}}
  & p \stacksymba{1 p X Y} &\ctran{1} p \stacksymba{\{p X Y\}} \\
 p \stacksymba{1 q Y} &\ctran{\frac23} p \stacksymba{\{q Y\}}
  & p \stacksymba{1 p X Y} &\ctran{\frac23} p \stacksymba{\{p X Y\}} \\
 p \stacksymba{1 q Y} &\ctran{\frac13} p \stacksymba{\{q Y\}}
  & p \stacksymba{1 p X Y} &\ctran{\frac13} p \stacksymba{\{p X Y\}} \\
p \stacksymba{\{q, p Y X\}} &\ctran{\#} q 
 & p \stacksymba{\{p, p X\}} &\ctran{\#} p \\
p \stacksymba{\{q, p Y X\}} &\ctran{\#} p Y X 
 & p \stacksymba{\{p, p X\}} &\ctran{\#} p X \\
p \stacksymba{\{q\}}        &\ctran{\#} q 
 & p \stacksymba{\{p\}} &\ctran{\#} p \\
p \stacksymba{\{p Y X\}}    &\ctran{\#} p Y X 
 & p \stacksymba{\{p X\}} &\ctran{\#} p X \\
p \stacksymba{\{q Y\}}      &\ctran{\#} q Y
 & p \stacksymba{\{p X Y\}} &\ctran{\#} p X Y
\end{aligned}
\end{gather*}
\caption{Set $\ctran{}$ of rules of the PDA $\Delta'$ from
Example~\ref{ex-gen-constr-PDA}}\label{fig:rulesofexamplepda}
\end{figure}
\end{exa}

\emph{Remark.} We note that the choice of $q_0$ plays no role. We
could
introduce a fresh ``don't care'' state, but the choice of $q_0$ from
$Q$ makes clear that the
original state set $Q$ need not be extended. Hence if $\Delta$ is a pBPA 
then $\Delta'$ is a BPA.

\medskip

Let $\calL=\calL_\Delta$ be the pLTS generated by a pPDA $\Delta$. 
The (standard) LTS 
$\calL_{\Delta'}$
defined by
$\Delta'$ is bigger than $\calL'$ as defined in
Section~\ref{sec-prob-to-nondet}, since 
the new stack symbols
 $\stacksymba{d}$ and  $\stacksymba{T}$ can occur anywhere in the
 stack in
 configurations of $\Delta'$, not only at the top as intended.
But if we restrict  $\calL_{\Delta'}$ to the states of 
$\calL_\Delta$ (i.e., to the configurations of $\Delta$), and close
this set under reachability (in $\calL_{\Delta'}$), then we obviously
get an LTS that is isomorphic with $\calL'$.
Hence by Lemma~\ref{lem:gen-nondet-constr-correct} we deduce 
the following theorem:

\begin{thm}\label{thm:prob-to-nondet}
For any pPDA
$\Delta$ there is a PDA $\Delta'$ constructible in exponential time such that
for any configurations $q_1\gamma_1, q_2\gamma_2$ of $\Delta$ we have
$q_1\gamma_1 \sim q_2\gamma_2$ in $\calL_\Delta$ if and only if $q_1\gamma_1
	\sim q_2\gamma_2$ in $\calL_{\Delta'}$.
Hence bisimilarity for pPDA is decidable.
\end{thm}

In Sections~\ref{sec-upper-bounds-visibly} 
and~\ref{sec-upper-bounds-pBPA} we will refer to the above
``stack-version'' PDA $\Delta'$ constructed to a given pPDA $\Delta$.
In Section~\ref{sec:bispOCAinPSPACE} we will use the following
``control-state version'' $\Deltac'$.

For a pPDA $\Delta = (Q,\Gamma,\Sigma,\mathord{\btran{}})$,
the PDA $\Deltac' =
(Q',\Gamma,\Sigma',\mathord{\ctran{}})$ arises analogously as $\Delta'$ but
with the following modifications: 
\begin{itemize}
	\item
The symbols $\stacksymba{d}$, $\stacksymba{T}$ are added to $Q$
instead of $\Gamma$. 
\item
Every
rule $qX \btran{a} d$ of $\Delta$ give rise to the following rules of
$\Deltac'$:

$qX\ctran{a}\stacksymba{d}X$;
$\stacksymba{d}X \ctran{\rho} \stacksymba{T}X$ (for $\rho$ and $T$ as
		in $\Delta'$); 
$\stacksymba{T}X \ctran{\#} p\alpha$ for $p\alpha \in T$.				\end{itemize}
We note that  
if $\calL=\calL_\Delta$ then $\calL_{\Deltac'}$ gets isomorphic
with the LTS $\calL'$ 
 defined in
Section~\ref{sec-prob-to-nondet} when we identify 
its states $\stacksymba{d}X\gamma$, $\stacksymba{d}Y\gamma$
and also $\stacksymba{T}X\gamma$, $\stacksymba{T}Y\gamma$.
In other words, in the configurations 
$\stacksymba{d}\alpha$, $\stacksymba{T}\alpha$ the first stack symbol
plays no role. We also note that $\Deltac'$ is an OCA when $\Delta$ is
a pOCA.

\medskip

\emph{Remark.} 
We will show that known algorithms for vPDA, BPA, OCA
in principle also work for pvPDA, pBPA, pOCA with no substantial complexity
increase, since they can be viewed as working on ``big'' PDA 
$\Delta'$ or $\Deltac'$
though these big PDA will be only implicitly presented by ``small'' pPDA
$\Delta$. We could try to define a more abstract notion 
of ``concise PDA'' that represent bigger standard PDA 
so that the algorithms for standard cases would also work for concise
PDA without a
substantial complexity increase, and
pvPDA, pBPA, pOCA would be just special cases of such concise PDA.
But we leave such an abstraction as a mere possibility, since it would add further technicalities.

\subsection{Bisimilarity of pvPDA is in EXPTIME} \label{sec-upper-bounds-visibly}

It is shown in~\cite[Theorem 3.3]{SrbaVisiblyPDA:2009} that the
bisimilarity problem for (standard) vPDA is EXPTIME-complete.
We will show that the same holds for pvPDA.
In this section we show the upper bound:

\begin{thm} \label{thm-upper-bounds-visibly}
The bisimilarity problem for pvPDA is in EXPTIME.
\end{thm}
In~\cite{SrbaVisiblyPDA:2009} the upper bound 
is proved by a reduction to the model-checking problem
 for (general) PDA and the modal $\mu$-calculus;
the latter problem is in EXPTIME by~\cite{Walukiewicz2001}.
This reduction does not apply in the probabilistic case, and if we
use
the reduction from Section~\ref{sec-prob-to-nondet} explicitly (as in 
Section~\ref{sec-pPDA-decidable}) and apply the result 
of~\cite{Walukiewicz2001} to the resulting 
exponentially bigger instance, 
then we only derive 
a double-exponential upper bound. In fact, if for a given pvPDA
$\Delta$ we construct the PDA $\Delta'$ as in 
Section~\ref{sec-pPDA-decidable}, then $\Delta'$ might be 
formally not a
vPDA (since the action $\#$ can have different effects on the stack
height), but it is straightforward to adjust the
construction so that $\Delta'$ becomes a vPDA.

Nevertheless,
we give a self-contained proof that the bisimilarity problem for pvPDA
is in EXPTIME; since vPDA is a special case, we thus also provide a
new proof of the EXPTIME upper bound in the standard case.

\begin{proof}
We consider a pvPDA 
$\Delta=(Q,\Gamma,\Sigma,\mathord{\btran{}})$ with the respective partition
 $\Sigma = \Sigmar \cup \Sigmai \cup \Sigmac$; it generates the
	respective pLTS $\calL_\Delta$.
	By $\Delta'$ we refer to the (standard) PDA arising as in 
Section~\ref{sec-pPDA-decidable},
but we do not assume constructing
it explicitly. We recall that $p\alpha\sim q\beta$ in $\calL_\Delta$
	iff $p\alpha\sim q\beta$ in $\calL_{\Delta'}$.

For each $a \in \Sigma$ we now define a relation $\mathord{\force{a}}$ (not necessarily a function) between pairs and sets of pairs:
	For $a\in\Sigma$, $p,q\in Q$, $X,Y\in\Gamma$, and $Out\subseteq 
		(Q\Gamma^{\leq 2})\times (Q\Gamma^{\leq
		2})$ we put
	$$(pX,qY)\force{a}Out$$
if the following conditions hold:
\begin{enumerate}
\item
	In the protocol running from $(pX,qY)$, i.e., in the respective
		three-round bisimulation game,
			Attacker can use an $a$-transition so that it
			is 
			 guaranteed
		that Defender either loses (having
			no response with action $a$) or the outcome is a pair
		from $Out$.	
	\item
If $a\in\Sigmac$, then $Out\subseteq 
		Q\Gamma\Gamma\times Q\Gamma\Gamma$;
if $a\in\Sigmai$, then $Out\subseteq 
		Q\Gamma\times Q\Gamma$;
if $a\in\Sigmar$, then $Out\subseteq 
		Q\times Q$.
\end{enumerate}			
We can easily verify 
that all relations $\mathord{\force{a}}$ can be constructed in exponential
time; indeed, to verify that $(pX,qY)\force{a}Out$ for concrete
$p,q,X,Y,a,Out$, it suffices to use
a polynomial-time algorithm with bounded
alternation (that simulates three rounds of the bisimulation game 
in $\calL_{\Delta'}$ starting in $(pX,qY)$).

Now by $\mathord{\forcelong}$ we denote the least relation
between the set $Q\Gamma\times Q\Gamma$ and 
	the set $2^{Q\times Q}$ (of subsets of $(Q\times Q)$) 
	satisfying the following (inductive) conditions:
\begin{enumerate}
\item
	if $(pX,qY)\force{a}Out$ for $a\in\Sigmar$ (hence $Out\subseteq 
		Q\times Q$),
		then $(pX,qY)\forcelong Out$;
\item
if 
$(pX,qY)\force{a}Out_1$ for $a\in\Sigmai$, 
		$Out_2\subseteq Q\times Q$,
		and 
	$(p'X',q'Y')\forcelong Out_2$ for each $(p'X',q'Y')\in Out_1$,
		then $(pX,qY)\forcelong Out_2$;
	\item
if 
$(pX,qY)\force{a}Out_1$ for $a\in\Sigmac$,
	$Out_2\subseteq Q\times Q$,
	and for each $(p'X_1X_2,q'Y_1Y_2)\in Out_1$ there is
		$Out'\subseteq Q\times Q$ such that 
		$(p'X_1,q'Y_1)\forcelong Out'$ and
		$(p''X_2,q''Y_2)\forcelong Out_2$ for each
		$(p'',q'')\in Out'$, then 
		 $(pX,qY)\forcelong Out_2$.
\end{enumerate}
It is obvious that $(pX,qY)\forcelong Out$ implies that Attacker has a strategy
guaranteeing that in each play
starting in $(pX,qY)$ he either wins or a pair
$(p'\varepsilon,q'\varepsilon)$ with $(p',q')\in Out$ is reached (when
he adheres to the strategy). The opposite implication also holds,
since if Attacker has such a strategy, then he can guarantee reaching
his goal in $n$ steps, for some $n\in\N$. (This follows from
image-finiteness, which entails that Defender always has a finite,
even bounded, number of possible responses.) 

To construct the set $\mathord{\forcelong}\subseteq 
(Q\Gamma\times Q\Gamma)\times 2^{Q\times Q}$, we can use 
the standard least-fixed-point computation:
we put $\mathord{\forcelong_0}=\emptyset$, and from  
$\mathord{\forcelong_i}$ we get ${\forcelong_{i+1}}\supseteq {\forcelong_{i}}$ 
by applying the above 
``deduction rules'' 1.--3.\ (replacing $\mathord{\forcelong}$ with 
$\mathord{\forcelong_i}$ in the
antecedents and with $\mathord{\forcelong_{i+1}}$ in the consequents).
This yields that $\mathord{\forcelong}$ can be constructed in exponential time
w.r.t.\ the size of the given pvPDA
$\Delta$.

For an exponential algorithm deciding if 
$p\alpha\sim q\beta$ for given configurations 
$p\alpha$, $q\beta$, we use an inductive construction of the sets
\begin{center}
$\calA(\alpha,\beta)=\{\,(p,q)\mid p\alpha\not\sim q\beta\,\}$. 
\end{center}
The construction is based on the following simple observations:
\begin{itemize}
\item		
$\calA(\varepsilon,\varepsilon)=\emptyset$;
\item
$\calA(\varepsilon,Y\beta')=\{\,(p,q)\mid qY$ enables some action 
$a\in\Sigma\,\}$;
\item
$\calA(X\alpha',Y\beta')=\{\,(p,q)\mid 
(pX,qY)\forcelong \calA(\alpha',\beta')\,\}$.
\end{itemize}
For deciding if $pX_1X_2\cdots X_m\sim qY_1Y_2\cdots Y_n$,
where we assume $m\leq n$ w.l.o.g., we can stepwise construct
the sets $\calA_m=\calA(\varepsilon,Y_{m+1}Y_{m+2}\cdots Y_n)$,
$\calA_{i-1}=\{\,(p',q')\mid (p'X_i,q'Y_i)\forcelong \calA_{i}\,\}$
for $i=m,m{-}1,\ldots,1$, and finally give the answer YES 
if $(p,q)\not\in \calA_0$, and NO if  $(p,q)\in \calA_0$.
\end{proof}

\begin{exa}\label{ex-vPDA}
Consider the pvPDA $\Delta = (\{p,p',q,q'\}, \{X\}, \{c,r\}, \mathord{\btran{}})$ with the following rules:
\[
p X \btran{c} \textstyle\frac12 p X X + \frac12 p' X X, \qquad p' X
\btran{c} p' X X, \qquad p X \btran{r} q, \qquad p' X \btran{r} q',
\qquad q X \btran{r} q.
\]
We have
\[
 (p X, p' X) \force{c} \{(p X X, p' X X)\}, \qquad (p X, p' X)
 \force{r} \{(q,q')\}, \qquad (q X, q' X) \force{r} \emptyset,
\]
and this is not an exhaustive list.
The first deduction rule yields $(p X, p' X) \forcelong \{(q,q')\}$ and  $(q X, q' X) \forcelong \emptyset$.
By the third deduction rule we thus have $(p X, p' X) \forcelong \emptyset$.
Hence $(p,p') \in \calA(X,X)$, and so $p X \not\sim p' X$.
\end{exa}

\subsection{Bisimilarity of pBPA is in
2-EXPTIME}\label{sec-upper-bounds-pBPA}

Bisimilarity of (standard) BPA is known to be in
2-EXPTIME (as claimed in~\cite{Burkart00} and explicitly proved
in~\cite{Jancar12}). Here we show that the same upper bound applies to
pBPA as well.

We first recall that if $\Delta$ is a pBPA then the pPDA $\Delta'$
defined in Section~\ref{sec-pPDA-decidable} is a BPA (since $\Delta$
and $\Delta'$ have the same singleton state sets).  Hence
Theorem~\ref{thm:prob-to-nondet}, with the 2-EXPTIME result for BPA,
immediately yields a triple-exponential upper bound for pBPA.  To argue that
this exponential increase is not necessary, we recall the proof
from~\cite{Jancar12} and show that mild modifications yield a
double-exponential algorithm also for pBPA. We assume the reader has
access to~\cite{Jancar12}, and we thus only recall the parts of the proof relevant
to the generalization rather than repeating the
whole argument in all detail. Nevertheless, we also try to convey the
intuitive ideas from~\cite{Jancar12} to facilitate
understanding. 

\begin{thm}\label{thm-upper-bounds-BPA}
The bisimilarity problem for pBPA is in 2-EXPTIME.
\end{thm}

\begin{proof}
We consider a pBPA $\Delta=(\Gamma,\Sigma,\mathord{\btran{}})$, omitting the
	singleton state set $Q$;
 it generates the
	respective pLTS $\calL_\Delta$.
	By $\Delta'$ we refer to the (standard) BPA arising as in 
Section~\ref{sec-pPDA-decidable},
but we do not assume constructing
it explicitly. We recall that 
	\begin{equation}\label{eq:pBPAandBPA}	
\textnormal{
	$\alpha\sim_n \beta$ in $\calL_\Delta$
	iff $\alpha\sim_{3n} \beta$ in $\calL_{\Delta'}$ 
}
	\end{equation}
		(using
	Lemma~\ref{lem:gen-nondet-constr-correct}).
One important fact is that the relations $\mathord{\sim_n}$ and $\mathord{\sim}$ are
congruences w.r.t.\ concatenation, i.e.: $\alpha\sim_n\alpha'$ and
 $\beta\sim_n\beta'$ imply $\alpha\beta\sim_n\alpha'\beta'$.
This holds in the LTS $\calL_{\Delta'}$	
by Proposition 3.1 in~\cite{Jancar12}, and thus also in 
	the pLTS $\calL_{\Delta}$ (due to~(\ref{eq:pBPAandBPA})),
when assuming that each 
$X\in\Gamma$ enables at least one action (cf.~the Remark after
Proposition 3.1 in~\cite{Jancar12}).
	
To match 
the approach 
in~\cite{Jancar12}, we assume that the states in $\calL_\Delta$ are
not only finite strings from $\Gamma^*$ but also
infinite regular (i.e., ultimately periodic) strings from
$\Gamma^{\omega}$.

In Section 3.2 of~\cite{Jancar12} the main algorithm is described
in the form of a Prover-Refuter game. To give an intuition, we
consider a pair of the form $(A\alpha,B\beta)$ for which Prover claims 
 $A\alpha\sim B\beta$ and Refuter  claims $A\alpha\not\sim B\beta$;
we have $A,B\in \Gamma$ and $\alpha,\beta\in\Gamma^*\cup\Gamma^\omega$.
(The notation in~\cite{Jancar12} uses $\calG=(\calN,\calA,\calR)$ for
BPA.)
By the (Prover-Refuter) game protocol, 
when $\Delta$ is a BPA,
Prover has the option to decide that one round of the
	standard bisimulation game, played by Attacker and
	Defender, will be mimicked: in this case Refuter performs an Attacker's
	move, 
	from the pair-component $A\alpha$ or $B\beta$, and Prover
	responds from the other pair-component by a move with the same
	action. 
	We thus get two moves $A\alpha\tran{a}\gamma_1\alpha$
	and $B\beta\tran{a}\gamma_2\beta$ where Prover claims 
$\gamma_1\alpha\sim \gamma_2\beta$ and Refuter claims that 
$\eqlevel(\gamma_1\alpha,\gamma_2\beta)<\eqlevel(A\alpha,B\beta)$,
where the \emph{equivalence-level} is defined as 
$\eqlevel(\delta_1,\delta_2)=\max\{n\mid \delta_1\sim_n\delta_2\}$.
The play then continues with the pair
$(\gamma_1\alpha,\gamma_2\beta)$.

The described option of one round of the bisimulation game corresponds
to Point (3c) in Section 3.2 of~\cite{Jancar12}.
Now, when $\Delta$ is a pBPA, the respective option is that 
Prover and Refuter mimic one round of the
	``probabilistic game'', i.e., of three rounds of the 
	bisimulation game 
	in $\calL_{\Delta'}$, which also results in a new pair 
	$(\gamma_1\alpha,\gamma_2\beta)$ of states from
	$\calL_\Delta$, as above.
	Such three rounds can be performed in polynomial time
	(w.r.t.\ the size of the pBPA $\Delta$)
	with bounded alternation. Since we discuss an 
	alternating-expspace complexity bound 
	(recalling that AEXPSPACE$=$2-EXPTIME),
	the described change of (3c) causes no problem.

Another option that Prover has for the pair $(A\alpha,B\beta)$ is to
provide several relevant smaller pairs from which Refuter will choose one to
continue. Concretely, Prover can decide to provide 
two or three smaller pairs in the forms captured by the following three 
possibilities:
\begin{enumerate}
\item
	$(\alpha,\gamma\beta)$, $(A\gamma,B)$;
\item
	$(\alpha,\gamma\beta)$, $(\beta,\delta^\omega)$,
	$(A\gamma\delta^\omega,B\delta^\omega)$;
\item
	$(\alpha,(\gamma\delta)^\omega)$, $(\beta,(\delta\gamma)^\omega)$,
	$(A(\gamma\delta)^\omega,B(\delta\gamma)^\omega)$.
\end{enumerate}
(As expected, $\delta^\omega=\delta\delta\delta\cdots$.)
In each of the above cases, $(A\alpha,B\beta)$ belongs to the least
congruence containing the respective (two or three) pairs; this
entails that the
least equivalence-level of the respective pairs cannot be bigger than 
 $\eqlevel(A\alpha,B\beta)$.
The size of finite strings from $\Gamma^*$ and regular strings from 
$\Gamma^\omega$
is based on the notion of norms.
By the \emph{norm} $\norm{A}$, for $A\in\Gamma$, we mean the length of a shortest $u$ such that 
$A\tran{u}\varepsilon$; if there is no such $u$, then $A$ is
\emph{unnormed} and we put  $\norm{A}=\omega$. 
We can imagine that we refer to the BPA $\Delta'$ when computing
$\norm{A}$. But since each (completed) path in $\calL_{\Delta'}$ that starts
from a state in $\calL_\Delta$ 
is composed of segments 
$X\alpha\tran{a}\stacksymba{d}\alpha\tran{\rho}\stacksymba{T}\alpha
\tran{\#}\gamma\alpha$ 
(where $X\btran{a}d$ is a rule of $\Delta$ and
$d(\gamma)>0$),
the
finite norms 
$\norm{A}$, $A\in\Gamma$, are at most exponential, and computable in
polynomial time,
\emph{w.r.t. the size of
the pBPA} $\Delta$
(as easily follows by mimicking the proof of Proposition 3.6.
in~\cite{Jancar12}).
We also put $\norm{\varepsilon}=0$ and
$\norm{A\alpha}=\norm{A}+\norm{\alpha}$, when $\norm{A}<\omega$.
Since $U\alpha\sim U$ when $U$ is unnormed, we only consider the
strings of the form $\alpha$, $\alpha U$, $\alpha(\beta)^\omega$
such that 
$\alpha$ and $\beta$ are normed strings from $\Gamma^*$ 
(i.e., $\norm{\alpha}$ and $\norm{\beta}$ are finite) and $U\in\Gamma$
is unnormed. In these cases we put
$\sizeofstr{\alpha}=\sizeofstr{\alpha U}=\norm{\alpha}$ and
$\alpha(\beta)^\omega=\sizeofstr{\alpha}+\sizeofstr{\beta}$ if 
$\alpha$, $\beta$ constitute the canonical presentation of the regular
infinite string $\alpha(\beta)^\omega$ (as defined in a standard way
in~\cite{Jancar12}).
We also put
$\sizeofstr{\delta_1,\delta_2}=\max\{\sizeofstr{\delta_1},\sizeofstr{\delta_2}\}$.

Lemma~3.14 in~\cite{Jancar12} shows that if
$\sizeofstr{A\alpha,B\beta}$ is bigger than an exponential constant
and $A\alpha\sim B\beta$,
then there is a decomposition of one of the above types 1,2,3 
consisting of two or three pairs of the type $(\delta_1,\delta_2)$
where $\delta_1\sim\delta_2$ and
$\sizeofstr{\delta_1,\delta_2}<\sizeofstr{A\alpha,B\beta}$.
We can apply the mentioned lemma to the BPA $\Delta'$;
since the respective exponential constant is related to maximal finite norms
$\norm{A}$, it is exponential in the size of the pBPA $\Delta$.

Hence exponential space is sufficient to keep the current game configuration so
that Prover has a winning strategy for $(\delta_1,\delta_2)$
iff $\delta_1\sim\delta_2$. 
(Prover wins a possibly infinite play if she is always able to respond
when Refuter is mimicking Attacker in the bisimulation game and
always keeps the game configuration in the determined exponential
bounds, by using a decomposition when a large pair arises.)

In the above conclusion, a subtle point is left implicit.
We apply Lemma~3.14 from~\cite{Jancar12} to a pair 
$(A\alpha,B\beta)$ of states in
$\calL_\Delta$ as to a pair of states in $\calL_{\Delta'}$, and assume
that we can choose the respective decompositions that consist of pairs of states
in $\calL_\Delta$. This cannot be deduced just from the statement of the
lemma, but it suffices to perform a routine check of the proofs of 
Lemma~3.14 and the related Lemma~3.11
from~\cite{Jancar12} to verify that such decompositions indeed exist.
The only point in the proofs that might be not so straightforward
is in (1) of
the proof of Lemma 3.11 in~\cite{Jancar12}. There we read
\begin{quote}
	We fix a rule $A_j\tran{a}\sigma_j$ such that for any rule
	$A_{3-j}\tran{a}\sigma_{3-j}$
	we get
	$\textsc{EqLv}(A_1\gamma,A_2)>\textsc{EqLv}(\sigma_1\gamma,\sigma_2)$
	$\ldots$ Now we fix a rule  $A_{3-j}\tran{a}\sigma_{3-j}$ such that
	$\sigma_1\gamma\mu_i\beta\sim \sigma_2\mu_i\beta$.
\end{quote}
Instead we now say: 
\begin{quote}
We fix a strategy of Attacker in the three-round
bisimulation game from $(A_1,A_2)$ such that the outcome
$(\sigma_1,\sigma_2)$ (for whatever strategy of Defender) will
satisfy $\textsc{EqLv}(A_1\gamma,A_2)>\textsc{EqLv}(\sigma_1\gamma,\sigma_2)$ 
$\ldots$ Now we fix a strategy of Defender, in the three-round
bisimulation game from $(A_1,A_2)$, such that the outcome
$(\sigma_1,\sigma_2)$ will satisfy $\sigma_1\gamma\mu_i\beta\sim \sigma_2\mu_i\beta$. \qedhere
\end{quote}
\end{proof}

\subsection{Bisimilarity
of pOCA is in PSPACE}\label{sec:bispOCAinPSPACE}

The bisimilarity problem for (standard)
one-counter automata is known to be PSPACE-complete.
We now argue that 
the decision algorithm described in~\cite{BGJ14} 
(originating in~\cite{BGJ:Concur10}) 
also shows, in fact,
that the problem for pOCA is in PSPACE as well.
We thus show:

\begin{thm}\label{thm-bisim-pOCA-inPspace}
The bisimilarity problem for pOCA is in PSPACE, even
if we present the instance
 $\Delta=(Q,\{I,Z\},\Sigma,\mathord{\btran{}})$,
$p I^mZ, q I^nZ$ (for which we ask if $p I^mZ\sim q I^nZ$)
by a shorthand using $m,n$ written in binary.
\end{thm}

\begin{proof}
Let us consider a pOCA $\Delta=(Q,\{I,Z\},\Sigma,\mathord{\btran{}})$.
	Let $\Deltac'=(Q',\{I,Z\},\Sigma',\mathord{\ctran{}})$ be the
	OCA defined 
in Section~\ref{sec-pPDA-decidable}
	(i.e., the 
	``control-state
	version'' of the respective PDA); 	
		hence 
	we have
	\begin{center}
	$p I^m Z\sim_n q I^{m'} Z$ in $\calL_\Delta$
		iff $p I^m Z\sim_{3n} q I^{m'} Z$ in
		$\calL_{\Deltac'}$. 
	\end{center}
We now aim to apply the algorithm 
	from~\cite{BGJ14} to $\Deltac'$ but using only (``small'') $\Delta$
	as a presentation of (``big'') $\Deltac'$.
If the algorithm was applied to $\Deltac'$ explicitly, it would
	construct a semilinear description of 
	the 
	mapping 
$\chi\colon\N\times\N\times (Q'\times Q')\rightarrow \{1,0\}$
such that 		
	\begin{center}
	$\chi(m,n,(p,q))=1$ iff $pI^mZ\sim qI^nZ$.
	\end{center}
But the set $Q'$ can be exponentially larger than the size of $\Delta$.
The idea is that 
	we simply let the algorithm compute 
just the restriction of $\chi$ to the domain
$\N\times\N\times (Q\times Q)$; it will turn out that 
a description of this restricted mapping can be computed
in polynomial space w.r.t.\ the size of $\Delta$.

Generally, for any relation  $R$ on $Q\times (\{I\}^*Z)$,
by the  \emph{(characteristic) colouring} $\chi_R$ we mean
the mapping $\chi_R\colon\N\times\N\times (Q\times Q)\rightarrow \{1,0\}$
where
	\begin{center}
$\chi_R(m,n,(p,q))=1$ iff $(pI^mZ,qI^nZ)\in R$.
	\end{center}
	In Fig.~\ref{fig:belts} we can see a depiction of the domain
	$\N\times\N\times (Q\times Q)$, assuming
	$Q=\{q_1,q_2,\dots,q_k\}$.
For a relation  $R$ on $Q\times (\{I\}^*Z)$,
the mapping $\chi_R$ can be viewed as a ``black-white colouring'',
	making a point $(m,n,(q_i,q_j))$ black if 
	$\chi_R(m,n,(q_i,q_j))=1$ and white
	if $\chi_R(m,n,(q_i,q_j))=0$.
	(See also the figures in~\cite{BGJ14}.)

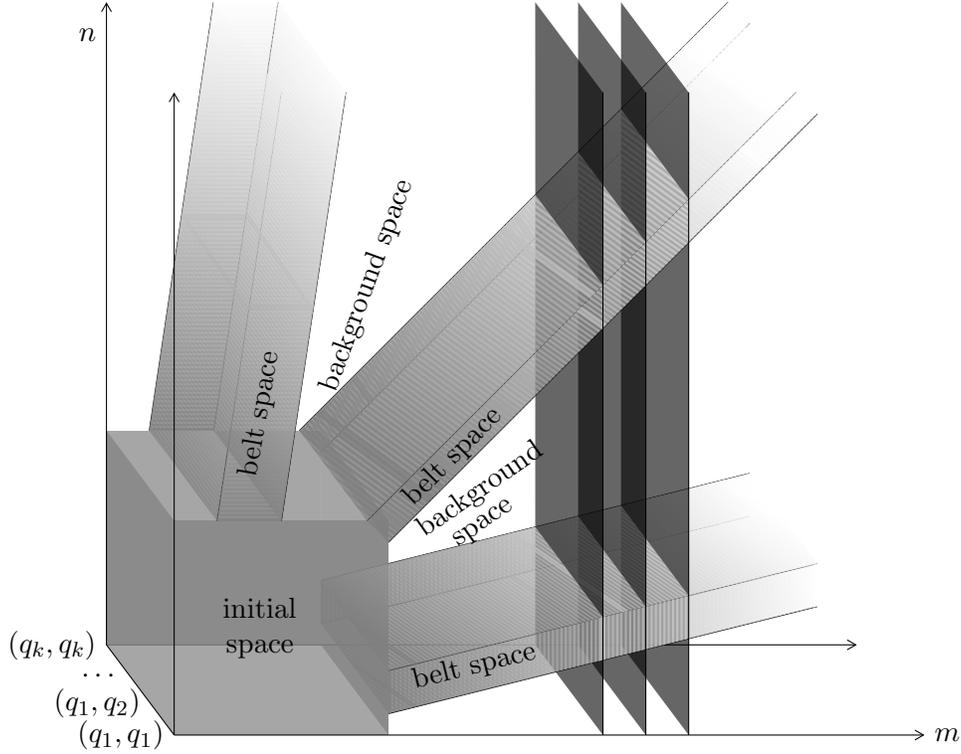
\begin{figure}
\centering
\begin{tikzpicture}[x = {(1.9cm,0cm)}, y = {(0cm,1.9cm)}, z = {(-0.6cm,0.8cm)},xscale=1.5,yscale=1.5]
 \path [name path=w1](2.0,0,0)--(2.0,3,0);
 \path [name path=w2](2.0,0,1)--(2.0,3,1);
 \path [name path=ww1](2.2,0,0)--(2.2,3,0);
 \path [name path=ww2](2.2,0,1)--(2.2,3,1);
 \path [name path=www1](2.4,0,0)--(2.4,3,0);
 \path [name path=www2](2.4,0,1)--(2.4,3,1);
 \path [name path=r1](1,0.3,0) -- (3,0.8,0);
 \path [name path=r2](1,0.3,1) -- (3,0.8,1);
 \path [name path=u1](0.9,1,0) -- (2.9,3,0);
 \path [name path=u2](0.9,1,1) -- (2.9,3,1);
 \path [name intersections={of= w1 and r1, name=A}];
 \path [name intersections={of= w2 and r2, name=B}];
 \path [name intersections={of= w1 and u1, name=C}];
 \path [name intersections={of= w2 and u2, name=D}];
 \path [name intersections={of= ww1 and r1, name=Aw}];
 \path [name intersections={of= ww2 and r2, name=Bw}];
 \path [name intersections={of= ww1 and u1, name=Cw}];
 \path [name intersections={of= ww2 and u2, name=Dw}];
 \path [name intersections={of= www1 and r1, name=Aww}];
 \path [name intersections={of= www2 and r2, name=Bww}];
 \path [name intersections={of= www1 and u1, name=Cww}];
 \path [name intersections={of= www2 and u2, name=Dww}];

 \draw[->] (0,0,1) -- (3.5,0,1);

 \draw[path fading=north] (1,0.1,1) -- (3,0.6,1);
 \draw[path fading=north] (1,0.3,1) -- (3,0.8,1);
 \shade[opacity=0.6, shading angle=90] (1,0.1,1) -- (1,0.3,1) -- (3,0.8,1) -- (3,0.6,1) -- cycle;
 \shade[opacity=0.6, shading angle=120] (1,0.1,0) -- (3,0.6,0) -- (3,0.6,1) -- (1,0.1,1) -- cycle;

 \draw[path fading=east] (0.9,1,1) -- (2.9,3,1);
 \draw[path fading=east] (1,0.9,1) -- (3,2.9,1);
 \shade[opacity=0.6, shading angle=125] (0.9,1,1) -- (1,0.9,1) -- (3,2.9,1) -- (2.9,3,1) -- cycle;
 \shade[opacity=0.6, shading angle=145] (1,0.9,0) -- (3,2.9,0) -- (3,2.9,1) -- (1,0.9,1) -- cycle;

 \path[fill=gray,opacity=0.7] (0,1,0) rectangle (1,0,0);
 \path[fill=gray,opacity=0.7] (0,1,1) rectangle (1,0,1);
 \path[fill=gray,opacity=0.7] (1,1,1) -- (1,1,0) -- (1,0,0) -- (1,0,1) -- cycle;
 \path[fill=gray,opacity=0.7] (0,1,1) -- (0,1,0) -- (0,0,0) -- (0,0,1) -- cycle;
 \node at (0.4,0.5,0) {\parbox{1.5cm}{\centering initial \\ space}};

 \draw[path fading=north] (0.2,1,0) -- (0.5,3,0);
 \draw[path fading=north] (0.2,1,1) -- (0.5,3,1);
 \draw[path fading=north] (0.5,1,0) -- (0.8,3,0);
 \draw[path fading=north] (0.5,1,1) -- (0.8,3,1);
 \shade[opacity=0.6, shading angle=160] (0.2,1,0) -- (0.5,3,0) -- (0.5,3,1) -- (0.2,1,1) -- cycle;
 \shade[opacity=0.6, shading angle=180] (0.2,1,0) -- (0.5,1,0) -- (0.8,3,0) -- (0.5,3,0) -- cycle;
 \shade[opacity=0.6, shading angle=160] (0.5,1,0) -- (0.8,3,0) -- (0.8,3,1) -- (0.5,1,1) -- cycle;
 \shade[opacity=0.6, shading angle=180] (0.2,1,1) -- (0.5,1,1) -- (0.8,3,1) -- (0.5,3,1) -- cycle;
 \node[rotate=80] at (0.4,1.5,0) {belt space};

 \path[opacity=0.6,fill] (2.4,0,0) -- (Aww-1) -- (Bww-1) -- (2.4,0,1) -- cycle;
 \path[opacity=0.6,fill] (2.2,0,0) -- (Aw-1) -- (Bw-1) -- (2.2,0,1) -- cycle;
 \path[opacity=0.6,fill] (2,0,0) -- (A-1) -- (B-1) -- (2,0,1) -- cycle;

 \draw[path fading=north] (1,0.1,0) -- (3,0.6,0);
 \draw[path fading=north] (1,0.3,0) -- (3,0.8,0);
 \shade[opacity=0.6, shading angle=120] (1,0.3,0) -- (3,0.8,0) -- (3,0.8,1) -- (1,0.3,1) -- cycle;
 \path[opacity=0.6,fill] (Aww-1) -- (Cww-1) -- (Dww-1) -- (Bww-1) -- cycle;
 \path[opacity=0.6,fill] (Aw-1) -- (Cw-1) -- (Dw-1) -- (Bw-1) -- cycle;
 \path[opacity=0.6,fill] (A-1) -- (C-1) -- (D-1) -- (B-1) -- cycle;
 \shade[opacity=0.6, shading angle=90] (1,0.1,0) -- (1,0.3,0) -- (3,0.8,0) -- (3,0.6,0) -- cycle;
 \node[rotate=14] at (1.4,0.32,0) {belt space};

 \draw[path fading=east] (0.9,1,0) -- (2.9,3,0);
 \draw[path fading=east] (1,0.9,0) -- (3,2.9,0);
 \shade[opacity=0.6, shading angle=145] (0.9,1,0) -- (2.9,3,0) -- (2.9,3,1) -- (0.9,1,1) -- cycle;
 \path[opacity=0.6,fill] (2.4,3,0) -- (Cww-1) -- (Dww-1) -- (2.4,3,1) -- cycle;
 \path[opacity=0.6,fill] (2.2,3,0) -- (Cw-1) -- (Dw-1) -- (2.2,3,1) -- cycle;
 \path[opacity=0.6,fill] (2,3,0) -- (C-1) -- (D-1) -- (2,3,1) -- cycle;
 \shade[opacity=0.6, shading angle=145] (0.9,1,0) -- (1,0.9,0) -- (3,2.9,0) -- (2.9,3,0) -- cycle;
 \node[rotate=45] at (1.3,1.3,0) {belt space};

 \draw[->] (0,0,1) -- (0,3,1) node[pos=0.95,left] {$n$};
 \draw[->] (0,0,0) -- (3.5,0,0) node[pos=1,right] {$m$};
 \draw[->] (0,0,0) -- (0,3,0);
 \draw (0,0,0) -- (0,0,1) node[pos=0.0,left] {$(q_1,q_1)$}
  node[pos=0.35,left] {$(q_1,q_2)$}
  node[pos=0.65,left] {$\cdots$}
  node[pos=1,left] {$(q_k,q_k)$};
 \node[rotate=70] at (0.9,2.1,0) {background space};
 \node[rotate=35] at (1.4,1.05,0) {\parbox{1.5cm}{\centering background\\[-1mm] space}};
 \draw (2,0,0) -- (2,3,0);
 \draw (2.2,0,0) -- (2.2,3,0);
 \draw (2.4,0,0) -- (2.4,3,0);
\end{tikzpicture}
\caption{OCA: partition of a grid, and a moving vertical
window~of~width~3}\label{fig:belts}
\end{figure}

An important ingredient in~\cite{BGJ14} is the underlying finite LTS;
here it underlies $\Deltac'$, and would be denoted 
by $\calF_{\Deltac'}$ according to Section 3.2 in~\cite{BGJ14}.
In our setting, we can describe $\calF_{\Deltac'}$ as follows.
We first consider the pLTS 
$\calL=(Q,\Sigma,\mathord{\tran{}})$ with the following relation
	$\mathord{\tran{}}$:
each rule $pI\btran{a}d$ of $\Delta$ gives rise to the 
transition $p\tran{a}d'$ 
where $d'(q)=d(q\varepsilon)+d(qI)+d(qII)$ (for all
$q\in Q$).	
Hence $\calL$ behaves as $\Delta$ with the ``always-positive''
counter. In other words, we can consider 
extending the pLTS $\calL_\Delta$ with the states $qI^\omega$ (for all $q\in
Q$);
then $\calL$ is just the restriction to the set of these additional states
(this set is closed under reachability).
Now the mentioned $\calF_{\Deltac'}$ is, in fact, the LTS 
$\calL'$ corresponding to $\calL$ as 
defined in Section~\ref{sec-prob-to-nondet}.

The bisimulation equivalence on the finite pLTS $\calL$ can be constructed in
polynomial time by standard partition refinement
technique~\cite{Baier96,DBLP:journals/jcss/BaierEM00}.
In particular, if $|Q|=k$, then ${\sim}={\sim_{k-1}}$ on $\calL$. 
Though $\calL'$ (i.e., $\calF_{\Deltac'}$) 
is exponentially bigger than $\calL$, 
due to its special form we have that 
\begin{center}
	${\sim}={\sim_{3(k-1)}}$ in $\calL'$.
\end{center}
Now the set $\INC$ from~\cite{BGJ14}
of configurations ``incompatible'' with $\calL'$
can be restricted to $Q$, and we thus put 
\begin{center}
	$\INC=\{\,\,pI^mZ\mid p\in Q, \forall q\in Q:pI^mZ\not\sim_{3k} q$
	when $pI^mZ$ is from $\calL_{\Deltac'}$ and $q$ from $\calL'\,\}$.
\end{center}
(We recall that comparing the states from different LTSs implicitly
refers to the disjoint union of these LTSs.)

As in~\cite{BGJ14}, 
we define $\distINC(pI^mZ)$ as the length of a shortest word
$u$ such that $pI^mZ\tran{u}qI^nZ\in\INC$ in the LTS $\calL_{\Deltac'}$;
we put $\distINC(pI^mZ)=\omega$ if there is no such word $u$.

The analysis
in~\cite{BGJ14} applied to $\Deltac'$, $\calL_{\Deltac'}$, $\calL'$
gives us the following:
\begin{enumerate}
	\item		
If $m\geq 3k$, then
		$pI^mZ\sim_{3k}p$ where
$p\in Q$, $pI^mZ$ is viewed as a state of $\calL_{\Deltac'}$,
		and $p$ is viewed as a state of $\calL'$.
		Hence $pI^mZ\in\INC$ implies $m<3k$. 
(In fact, $3k$ can be replaced by $k$ here, 
		since for any segment
		$pI^mZ\tran{a}\stacksymba{d}I^mZ\tran{\rho}
		\stacksymba{T}I^mZ\tran{\#}qI^{m'}Z$ 	
		we have $m'\geq m{-}1$,
		but this is not
		important.)
	\item
If $\distINC(pI^mZ)=\omega$, then  
$pI^mZ\sim q$ for some $q$ (even if $m<3k$), and, moreover,
$pI^mZ\sim q$ iff $pI^mZ\sim_{3k} q$.
\item
If $\distINC(pI^mZ)\neq \distINC(qI^nZ)$, then 
$pI^mZ\not\sim qI^nZ$.
\end{enumerate}
We note that $\INC$ can be computed in polynomial time 
w.r.t.\ the size of $\Delta$ (though even polynomial space would
suffice for us): by Point 1 we have $\INC\subseteq\{\,pI^mZ\mid p\in Q,
m<3k\,\}$, and deciding if $pI^mZ\sim_{3k} q$ can be easily done in
polynomial time.

When computing $\distINC(pI^mZ)$, we can consider only
easily computable ``macrosteps'' 
$pI^mZ\tran{}qI^{m'}Z$, with $p,q\in Q$ and $m'\in\{m{-}1,m,m{+}1\}$, where 
there are $a$,$\stacksymba{d}$,$\rho$,$\stacksymba{T}$ such that
$pI^mZ\tran{a}\stacksymba{d}I^mZ\tran{\rho}\stacksymba{T}I^mZ
\tran{\#}qI^{m'}Z$ in $\calL_{\Deltac'}$.

Hence the analysis in~\cite{BGJ14} shows that 
the points $(m,n,(p,q))$ such that $\distINC(pI^mZ)=
\distINC(qI^nZ)<\omega$ lie in ``linear belts'' 
(see Proposition $26$ in~\cite{BGJ14}), i.e., in the ``belt space''
and the ``initial space''
depicted in Figure~\ref{fig:belts}. Moreover, the belts have polynomial
coefficients in the size of $\Delta$:
though formally we refer to
$\Delta'$ with $|Q'|$ states, the coefficients are polynomial in the
number $|Q|$ due to the above mentioned small number of macrosteps. 

The polynomial space algorithm in~\cite{BGJ14} is ``moving the
vertical window of width $3$'' from the beginning to the right (as is
also depicted in Figure~\ref{fig:belts} and in the figures in~\cite{BGJ14}). 
In principle it can guess the black points in the initial space and
the belt space, while the colour of the points in the background space 
(where we either have
$\distINC(pI^mZ)\neq \distINC(qI^nZ)$
or $\distINC(pI^mZ)=\distINC(qI^nZ)=\omega$) can be easily computed.

The algorithm now proceeds in the same way as
\textsc{Alg-Bisim} in Section $4$ of~\cite{BGJ14};
it guesses a black-white coloring that should represent a
bisimulation 
but, as already mentioned, it guesses just 
the restriction to the domain
$\N\times\N\times(Q\times Q)$. 
In Point (c)ii of \textsc{Alg-Bisim}
it is checked if the guess in the intersection of the 
middle of the vertical window with the initial space and belt space
is consistent w.r.t.\ the black-white colouring of its (closest) neighbourhood.

In our case,
in such a black point $(m,n,(p,q))$ we simply run a polynomial-time algorithm
with bounded alternation, mimicking three rounds of the bisimulation
game from $(pI^mZ,qI^nZ)$, to check if Defender can guarantee that the
outcome is again a black point.
\end{proof}

\section{Lower Bounds}\label{sec-lower-bounds}

The upper bounds for pOCA 
and for
pvPDA are tight:
the bisimilarity problems already for standard versions are known to
be PSPACE-hard for 
OCA (even for visibly OCA~\cite{SrbaVisiblyPDA:2009})
and EXPTIME-hard for vPDA~(see~\cite{SrbaVisiblyPDA:2009} where 
a relevant construction from~\cite{KuceraM02} is used);
hence in combination with Theorems~\ref{thm-bisim-pOCA-inPspace} 
and~\ref{thm-upper-bounds-visibly} we obtain:

\begin{cor}
The bisimilarity problem for pOCA is PSPACE-complete, and 
the bisimilarity problem for pvPDA EXPTIME-complete.
\end{cor}
In Sections~\ref{sec:pOCA-hard} and~\ref{sec:pvPDA-hard}
we show that these lower bounds also apply to the fully
probabilistic versions, even when the action alphabet is restricted to size 1 and 3, respectively.
We define two gadgets, adapted from~\cite{CBW12}, that will be
used for both results.  The gadgets are small pLTSs that allow us to simulate AND and
OR gates using probabilistic bisimilarity. We depict the gadgets in
Figure~\ref{fig:gadgets}, where we assume that
all edges have probability $\frac12$ and have the same label. The gadgets satisfy
the following (trivially verifiable) propositions, in which we write
$s \tran{a}t_1\mid t_2$ as a shorthand for
$s\tran{a} \frac12 t_1 + \frac12 t_2$.

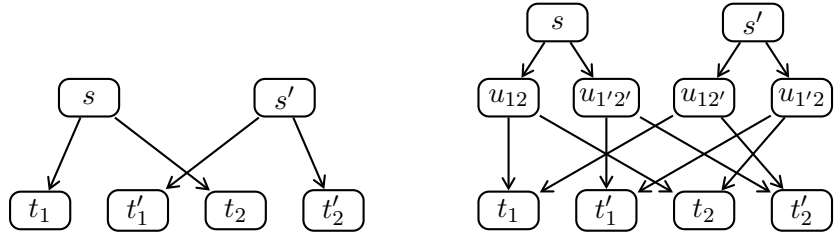
\begin{figure}
\begin{center}
\begin{tabular}{c@{\hspace{15mm}}c}
\begin{tikzpicture}[xscale=1.3,tran]
\tikzstyle{every node} = [ran,inner sep=2pt,minimum height=5mm, minimum width=8mm];
\node (aa) at (0.5,1.5)  {$s$};
\node (bb) at (2.5,1.5)  {$s'$};
\node (cc) at (0,0)  {$t_1$};
\node (dd) at (1,0)  {$t_1'$};
\node (ee) at (2,0)  {$t_2$};
\node (ff) at (3,0)  {$t_2'$};
\draw[->] (aa)--(cc);
\draw[->] (bb)--(dd);
\draw[->] (aa)--(ee);
\draw[->] (bb)--(ff);
\end{tikzpicture}
&
\begin{tikzpicture}[xscale=1.3,tran]
\tikzstyle{every node} = [ran,inner sep=2pt,minimum height=5mm, minimum width=8mm];
\node (a)  at (0.5,3)  {$s$};
\node (b) at (2.5,3)  {$s'$};
\node (c) at  (0,2)  {$u_{12}$};
\node (d) at  (1,2)  {$u_{1'2'}$};
\node (e) at  (2,2)  {$u_{12'}$};
\node (f) at  (3,2)  {$u_{1'2}$};
\node (g) at  (0,0.5)  {$t_1$};
\node (h) at  (1,0.5)  {$t_1'$};
\node (i) at  (2,0.5)  {$t_2$};
\node (j) at  (3,0.5)  {$t_2'$};
\draw[->]  (a) -- (c);
\draw[->]  (a) -- (d);
\draw[->]  (b) -- (e);
\draw[->]  (b) -- (f);
\draw[->]  (c) -- (g);
\draw[->]  (d) -- (h);
\draw[->]  (e) -- (j);
\draw[->]  (f) -- (i);
\draw[->]  (e) -- (g);
\draw[->]  (f) -- (h);
\draw[->]  (c) -- (i);
\draw[->]  (d) -- (j);
\end{tikzpicture}
\end{tabular}
\end{center}
\caption{AND-gadget (left) and OR-gadget (right)}\label{fig:gadgets}
\end{figure}

\begin{prop}\label{prop:ANDgadget}
\textnormal{(AND-gadget)}
Suppose $s,s'$,
$t_1,t'_1$, $t_2,t'_2$ are states in a pLTS such that
 $t_1\not\sim t'_2$ and the only transitions
outgoing from $s,s'$ are
$s\tran{a}t_1\mid t_2$
and
$s'\tran{a}t'_1\mid t'_2$\,.
Then
$s\sim s'$ if and only if $t_1\sim t'_1$ $\land$
$t_2\sim t'_2$.
\end{prop}

\begin{prop}\label{prop:ORgadget}
\textnormal{(OR-gadget)}
Suppose $s,s'$,
$t_1,t'_1$, $t_2,t'_2$,
and $u_{12}$, $u_{1'2}$, $u_{12'}$, $u_{1'2'}$
are states in a pLTS.
Let the only transitions
outgoing from
$s,s',u_{12}, u_{1'2},u_{12'},u_{1'2'}$
be
\begin{gather*}
s\tran{a}u_{12}\mid u_{1'2'} \qquad
s'\tran{a}u_{12'}\mid u_{1'2}
\\
u_{12}\tran{a}t_1\mid t_2 \qquad
u_{1'2'}\tran{a}t'_1\mid t'_2 \qquad
u_{12'}\tran{a}t_1\mid t'_2 \qquad
u_{1'2}\tran{a}t'_1\mid t_2\,.
\end{gather*}
Then
$s\sim s'$ if and only if $t_1\sim t'_1$ $\lor$
$t_2\sim t'_2$.
\end{prop}

\subsection{Bisimilarity of pOCA is PSPACE-hard}\label{sec:pOCA-hard}

In this section we prove the following:
\begin{thm} \label{thm-pOCA-hard}
	Bisimilarity is PSPACE-hard (even) for
 unary (i.e., with only one action) and fully probabilistic
	OCA.
\end{thm}
We remark that fully probabilistic PDA with only one action do not have any nondeterminism: they generate a countable-state Markov chain.
\begin{proof}
We use a reduction from the emptiness problem for alternating
finite automata with a one-letter alphabet, known to be
PSPACE-complete~\cite{Holzer95,JancarSawaAFA:2007};
our reduction resembles the reduction in~\cite{SrbaVisiblyPDA:2009}
for (non-probabilistic) visibly one-counter automata.

A \emph{one-letter alphabet alternating finite automaton}, 1L-AFA, is a tuple
$A=(Q,\delta,q_0,F)$ where
$Q$ is the (finite) set of \emph{states},
 $q_0$
is the \emph{initial state}, $F\subseteq Q$ is the set of
\emph{accepting states}, and the \emph{transition
function} $\delta$ assigns to each $q\in Q$
either $q_1\land q_2$ or $q_1\lor q_2$, where $q_1,q_2\in Q$.

We define the predicate $\Acc\subseteq Q\times\N$
by induction on the second component (i.e., the length of a one-letter word);
$\Acc(q,n)$ means
``$A$ starting in $q$ accepts $n$'':
$\Acc(q,0)$ if and only if $q\in F$;
$\Acc(q,n{+}1)$ if and only if
either
$\delta(q)=q_1\land q_2$ and
we have both $\Acc(q_1,n)$ and $\Acc(q_2,n)$, or
 $\delta(q)=q_1\lor q_2$ and we have
$\Acc(q_1,n)$ or $\Acc(q_2,n)$.

The \emph{emptiness problem for 1L-AFA} asks, given a 1L-AFA $A$, if the set
$\{\,n\mid \Acc(q_0,n)\,\}$ is empty.
We now reduce this problem to our problem.

Assuming a 1L-AFA $(Q,\delta,q_0,F)$, we construct a pOCA
	$\Delta=(\overline{Q},\{I,Z\},\{a\},\mathord{\btran{}})$ as
	follows. The state set $\overline{Q}$ contains
	$2|Q|+3$ `basic' states; the set of basic states
is $\{p_0,p'_0,r\}\cup Q\cup Q'$ where
$Q'=\{\,q'\mid q\in Q\,\}$ is a copy of $Q$ and $r$ is a special dead
state.
Additional auxiliary states will be added to implement AND- and
OR-gadgets. We note that $\Delta$ has a singleton action alphabet,
	and it will 
be fully probabilistic.
	Below we describe a construction of $\mathord{\btran{}}$, aiming
to achieve $p_0IZ\sim p_0'IZ$ if and only if $\{\,n\mid
	\Acc(q_0,n)\,\}=\emptyset$; another property will be that
\begin{equation}\label{eq:afapoca}
qI^{n}Z\sim q'I^{n}Z \textnormal{ if and only if } \neg\Acc(q,n).
\end{equation}
For each $q\in F$ we create a rule
$qZ \btran{a} r Z$,
but $qZ$ is dead (i.e., there is no rule $qZ \btran{a}\ldots$) if
$q\not\in F$;  $q'Z$ is dead for any $q'\in Q'$.
Both $rI$ and $rZ$ are dead as well.
Hence
(\ref{eq:afapoca}) is
satisfied for $n=0$.
Now we complete the set $\mathord{\btran{}}$ of rules 
and show that (\ref{eq:afapoca}) also holds for
$n>0$.

For $q\in Q$ with $\delta(q)=q_1\lor q_2$ we implement an
AND-gadget from Figure~\ref{fig:gadgets} (left) guaranteeing
that $qI^{n{+}1}Z\sim q'I^{n{+}1}Z$ if and only if
$q_1I^{n}Z\sim q_1'I^{n}Z$ and $q_2I^{n}Z\sim q_2'I^{n}Z$
(since $\neg\Acc(q,n{+}1)$ if and only if $\neg\Acc(q_1,n)$ and
$\neg\Acc(q_2,n)$):

We add rules
$qI \tran{} r_1I \mid r_2I$ (this is a shorthand for
$qI \btran{a} \frac12 r_1 I + \frac12 r_2 I$)
and $q'I \tran{} r'_1I \mid r'_2I$,
and also $r_1I \tran{} q_1 \mid s_1I$,
$r_2I \tran{} q_2 \mid s_2I$,
 $r'_1I \tran{} q'_1 \mid s_1I$,
$r'_2I \tran{} q'_2 \mid s_2I$,
and finally $s_1I\btran{a} \frac12 s_1I + \frac12 r$,
$s_2I\btran{a}0.4 s_2I + 0.6 r$.
The intermediate states $r_1,r_2,r'_1,r'_2$, and $s_1,s_2$ serve to
implement the condition
$t_1\not\sim t'_2$ from Proposition~\ref{prop:ANDgadget}.

For $q\in Q$ with $\delta(q)=q_1\land q_2$ we (easily) implement an
OR-gadget from Figure~\ref{fig:gadgets} (right)
guaranteeing $qI^{n{+}1}Z\sim q'I^{n{+}1}Z$ if and only if
$q_1I^{n}Z\sim q_1'I^{n}Z$ or $q_2I^{n}Z\sim q_2'I^{n}Z$
(since $\neg\Acc(q,n{+}1)$ if and only if $\neg\Acc(q_1,n)$ or
$\neg\Acc(q_2,n)$).

To finish the construction, we add rules
$p_0I\btran{a} \frac13 p_0 I I + \frac13 q_0 \varepsilon + \frac13 r I$ and
$p_0'I\btran{a} \frac13 p_0' I I + \frac13 q'_0 \varepsilon + \frac13 r I$.
As $p_0 I$ and $p_0' I$ can transition to (the dead) $r I$, the rules added before guarantee that
$p_0I^{n+2}Z\not\sim q'_0I^nZ$ and
$q_0I^{n}Z\not\sim p_0'I^{n+2}Z$.
\end{proof}

\begin{exa} \label{ex-pOCA-hardness}
We illustrate this reduction for the 1L-AFA $(A = \{q_0, q_1, q_2\},\delta,q_0,\{q_2\})$ with $\delta(q_0) = q_1 \land q_2$ and $\delta(q_1) = q_1 \lor q_2$ and $\delta(q_2) = q_1 \lor q_1$.
Here is a visualization of~$A$:
\begin{center}
\begin{tikzpicture}[tran,scale=1.5]
\tikzstyle{every node} = [ran,inner sep=2pt,minimum height=5mm, minimum width=8mm];
\node (q0) at (-0.2,0) {$q_0$};
\node (q1) at (1,0.5) {$q_1$};
\node[double] (q2) at (1,-0.5) {$q_2$};
\draw (-1,0) -- (q0);
\draw[rounded corners=2mm] (q0) to (0.4,0) to (q1);
\draw[rounded corners=2mm] (q0) to (0.4,0) to (q2);
\draw (q2) edge[bend right=15] (q1);
\draw (q1) edge[bend right=15] (q2);
\draw (q1) edge [loop,out=20,in=-20,looseness=10] (q1);
\end{tikzpicture}
\end{center}
Figures \ref{fig-ex-pOCA-hardness1} and \ref{fig-ex-pOCA-hardness2} show parts of the pLTS generated by the pOCA obtained by applying the reduction from the proof of Theorem~\ref{thm-pOCA-hard} to~$A$.
\begin{figure}
\begin{center}
\begin{tikzpicture}[xscale=1.9,yscale=1.3,thick]
\tikzstyle{lran} = [ran,inner sep=2pt,minimum height=5mm, minimum width=8mm];

\node[lran] (q1I)   at (1.5,2){$q_1 I Z$};
\node[lran] (r1I)   at (0.5,1){$r_1 I Z$};
\node[lran] (r2I)   at (2.5,1){$r_2 I Z$};
\node[lran] (q1)    at (0, 0) {$q_1 Z$};
\node[lran] (s1I)   at (1, 0) {$s_1 I Z$};
\node[lran] (q2)    at (2, 0) {$q_2 Z$};
\node[lran] (s2I)   at (3, 0) {$s_2 I Z$};

\node[lran] (q1'I)   at (5.5,2){$q_1' I Z$};
\node[lran] (r1'I)   at (4.5,1){$r_1' I Z$};
\node[lran] (r2'I)   at (6.5,1){$r_2' I Z$};
\node[lran] (q1')   at (4, 0) {$q_1' Z$};
\node[lran] (s1I')  at (5, 0) {$s_1 I Z$};
\node[lran] (q2')   at (6, 0) {$q_2' Z$};
\node[lran] (s2I')  at (7, 0) {$s_2 I Z$};

\node[lran] (rZ)    at (2,-1) {$r Z$};
\node[lran] (rZ')   at (6,-1) {$r Z$};

\draw[tran] (q1I) edge (r1I) edge (r2I);
\draw[tran] (r1I) edge (q1) edge (s1I);
\draw[tran] (r2I) edge (q2) edge (s2I);

\draw[tran] (s1I) edge [loop,out=-110,in=-70,looseness=20] (s1I);
\draw[tran] (s1I) -- (rZ);
\draw[tran] (q2)  -- (rZ);
\draw[tran] (s2I) edge node[pos=0.4,below,xshift=1mm] {$0.6$} (rZ);
\draw[tran] (s2I) edge [loop,out=-110,in=-70,looseness=20] node[below] {$0.4$} (s2I);

\draw[tran] (q1'I) edge (r1'I) edge (r2'I);
\draw[tran] (r1'I) edge (q1') edge (s1I');
\draw[tran] (r2'I) edge (q2') edge (s2I');

\draw[tran] (s1I') edge [loop,out=-110,in=-70,looseness=20] (s1I');
\draw[tran] (s1I') -- (rZ');
\draw[tran] (s2I') edge node[pos=0.4,below,xshift=1mm] {$0.6$} (rZ');
\draw[tran] (s2I') edge [loop,out=-110,in=-70,looseness=20] node[below] {$0.4$} (s2I');

\end{tikzpicture}
\end{center}
\caption{A part of the pLTS generated by the pOCA obtained by applying the reduction from the proof of Theorem~\ref{thm-pOCA-hard} to the 1L-AFA $A$ from Example~\ref{ex-pOCA-hardness}.
Unless indicated otherwise, for each state there is a uniform distribution on the outgoing transitions.
For better readability, some states appear twice.
}
\label{fig-ex-pOCA-hardness1}
\end{figure}
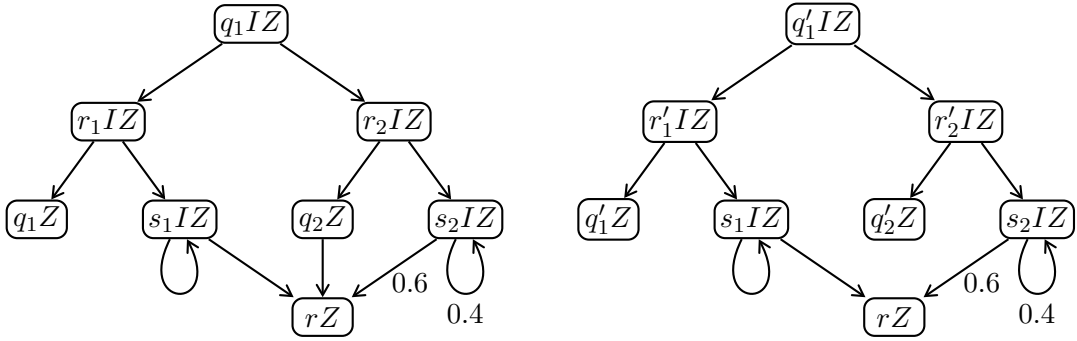
It can be seen in Figure~\ref{fig-ex-pOCA-hardness1} that $q_1 Z \sim q_1' Z$ and $q_2 Z \not\sim q_2' Z$, reflecting the facts that $\neg\Acc(q_1,0)$ and $\Acc(q_2,0)$.
Therefore $r_1 I Z \sim r_1' I Z$ and $r_2 I Z \not\sim r_2' I Z$.
Corresponding to the transition $\delta(q_1) = q_1 \lor q_2$, Figure~\ref{fig-ex-pOCA-hardness1} shows an AND-gadget, see Figure~\ref{fig:gadgets} (left).
Thus we have $q_1 I Z \not\sim q_1' I Z$, reflecting the fact that $\Acc(q_1,1)$.

\begin{figure}
\begin{center}
\begin{tikzpicture}[xscale=1.9,yscale=1.3,thick]
\tikzstyle{lran} = [ran,inner sep=2pt,minimum height=5mm, minimum width=8mm];

\node[lran] (q0IZ)   at (0.5,2) {$q_0 I Z$};
\node[lran] (q0'IZ)  at (2.5,2) {$q_0' I Z$};

\node[lran] (1)      at ( 0, 1) {};
\node[lran] (2)      at ( 1, 1) {};
\node[lran] (3)      at ( 2, 1) {};
\node[lran] (4)      at ( 3, 1) {};

\node[lran] (q1Z)    at ( 0, 0) {$q_1 Z$};
\node[lran] (q1'Z)   at ( 1, 0) {$q_1' Z$};
\node[lran] (q2Z)    at ( 2, 0) {$q_2 Z$};
\node[lran] (q2'Z)   at ( 3, 0) {$q_2' Z$};
\node[lran] (rZ)     at ( 2,-1) {$r Z$};

\draw[tran] (q0IZ)  edge (1) edge (2);
\draw[tran] (q0'IZ) edge (3) edge (4);

\draw[tran] (1) edge (q1Z) edge (q2Z);
\draw[tran] (2) edge (q1'Z) edge (q2'Z);
\draw[tran] (3) edge (q1Z) edge (q2'Z);
\draw[tran] (4) edge (q1'Z) edge (q2Z);

\draw[tran] (q2Z) -- (rZ);
\end{tikzpicture}
\end{center}
\caption{A part of the pLTS generated by the pOCA obtained by applying the reduction from the proof of Theorem~\ref{thm-pOCA-hard} to the 1L-AFA $A$ from Example~\ref{ex-pOCA-hardness}.
For each state there is a uniform distribution on the outgoing transitions.
}
\label{fig-ex-pOCA-hardness2}
\end{figure}
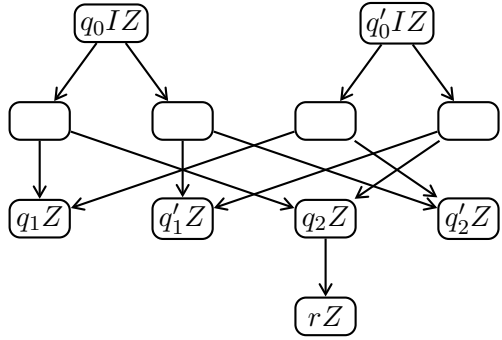
In Figure~\ref{fig-ex-pOCA-hardness2} it can be seen again that $q_1 Z \sim q_1' Z$ and $q_2 Z \not\sim q_2' Z$, reflecting the facts that $\neg\Acc(q_1,0)$ and $\Acc(q_2,0)$.
Corresponding to the transition $\delta(q_0) = q_1 \land q_2$, Figure~\ref{fig-ex-pOCA-hardness1} shows an OR-gadget, see Figure~\ref{fig:gadgets} (right).
Thus we have $q_0 I Z \sim q_0' I Z$, reflecting the fact that $\neg\Acc(q_0,1)$.
\end{exa}

\subsection{Bisimilarity of pvPDA is EXPTIME-hard}\label{sec:pvPDA-hard}

In this section we prove the following:
\begin{thm}\label{thm-pvPDA-hard}
 Bisimilarity is EXPTIME-hard
	(even) for fully probabilistic pvPDA 
	$\Delta=(Q,\Gamma,\Sigma,\mathord{\btran{}})$
	with $|\Sigmar| = |\Sigmai| = |\Sigmac| = 1$.
\end{thm}
It was shown in~\cite{SrbaVisiblyPDA:2009} that bisimilarity for (non-probabilistic) vPDA is EXPTIME-complete.
The hardness result there follows by observing that the proof given 
in~\cite{KuceraM02}
(see also~\cite{KuceraM10}) for general PDA works even for vPDA. 
Referring to~\cite{KuceraM02}, 
 it is commented in~\cite{SrbaVisiblyPDA:2009}: ``Though conceptually elegant, the technical details of the reduction are rather tedious.''
For those reasons we give a full reduction from the problem of
determining the winner in a reachability game on pushdown processes;
this problem was shown EXPTIME-complete in~\cite{Walukiewicz2001}.
Our reduction proves Theorem~\ref{thm-pvPDA-hard}
and at the same time provides an alternative proof for (standard)
vPDA.
\begin{proof}[Proof of Theorem~\ref{thm-pvPDA-hard}]
Let $\Delta = (Q,\Gamma,\{a\},\mathord{\btran{}})$ be a unary PDA
	(the actions do not matter in reachability games)	
with
 a control state partition $Q = Q_0 \cup Q_1$ and an initial configuration $p_0 X_0$.
We call a configuration $p X \alpha$ \emph{dead} if it has no successor configuration,
 i.e., if $\Delta$ does not have a rule with $p X$ on the left-hand side.
Consider the following game between Player~0 and Player~1 on the 
	LTS~$\calL_{\Delta}$:
First the current configuration is $p_0 X_0$;
in a current configuration $p \alpha$ with $p \in Q_i$ (where $i \in \{0,1\}$),
 Player~$i$ chooses a successor configuration of~$p \alpha$ 
	in~$\calL_{\Delta}$ as a new current configuration.
The goal of Player~$1$ is to reach a dead configuration; the goal of Player~$0$ is to avoid that.
It is shown in~\cite[pp.~261--262]{Walukiewicz2001} that determining the winner in that game is EXPTIME-hard.

W.l.o.g.\ we assume that for each 
	$pX\in Q\Gamma$ there are \emph{at most} two rules in $\mathord{\btran{}}$ with 
	$pX$ on the left-hand side; moreover, if they are two, then
	they are of the form $pX\btran{a}p_1X_1$ and 
	$pX\btran{a}p_2X_2$ (for $X_1,X_2\in\Gamma$).
We further assume that no configuration with the empty stack is
	reachable from $p_0X_0$.

	We will construct a fully probabilistic vPDA $\bar \Delta =
	(\bar Q, \Gamma, \{\ar, \ai, \ac\}, \mathord{\ctran{}})$
	such that for each control state $p \in Q$ the set~$\bar Q$
	includes $p$ and a copy~$p'$, and
	the configurations $p_0 X_0$ and~$p_0' X_0$
	of~$\bar \Delta$ are bisimilar if and only if Player~$0$
	has a winning strategy (in the reachability game from
	$p_0X_0$):
	\begin{itemize}	
		\item
			For each $p X \in Q\Gamma$ that is dead in~$\Delta$, in $\bar \Delta$ we 
	create a rule $p X \ctran{\ai} p X$ 
 and a rule $p' X \ctran{\ai} z X$ where $z \in \bar Q$ is a special control state not occurring on any left-hand side.
This ensures that if $p X$ is dead in~$\Delta$ (and hence Player~$1$
	wins), then we have $p X \not\sim p' X$ in~$
	\calL_{\bar \Delta}$.
\item
	For each rule $p X\btran{a}q\alpha$ 
	such that $pX$ is not the left-hand side of any other rule in
	$\Delta$, in $\bar \Delta$
 we create rules $p X \ctran{a} q \alpha$ and $p' X \ctran{a} q' \alpha$,
 where $a = \ar, \ai, \ac$ if $|\alpha| = 0, 1, 2$, respectively.
\item
	For each pair of (different) rules 
			$p X\btran{a}p_1X_1$,	$p X\btran{a}p_2X_2$:
\begin{itemize}
 \item If $p \in Q_0$, then we implement an OR-gadget from
	 Figure~\ref{fig:gadgets} (right):
\\		
  let $(p_1X_1p_2X_2), (p_1'X_1p_2'X_2), (p_1X_1p_2'X_2), (p_1'X_1p_2X_2) \in \bar Q$ be fresh control states,
  and add rules
		\begin{itemize}		
			\item
				$p X \ctran{} (p_1X_1p_2X_2) X \mid (p_1'X_1p_2'X_2) X$
\\				
   (this is a shorthand for $p X \ctran{\ai} 0.5 (p_1X_1p_2X_2) X + 0.5 (p_1'X_1p_2'X_2)
				X$), 
			\item
  $p' X \ctran{} (p_1X_1p_2'X_2) X
				\mid (p_1'X_1p_2X_2) X$,
			\item
   $(p_1X_1p_2X_2) X \ctran{} p_1 X_1 \mid p_2 X_2$,
   $(p_1'X_1p_2'X_2) X \ctran{} p_1' X_1 \mid p_2' X_2$,
  \\
				$(p_1X_1p_2'X_2) X \ctran{} p_1 X_1 \mid p_2' X_2$,
   $(p_1'X_1p_2X_2) X \ctran{} p_1' X_1 \mid p_2 X_2$.
	\end{itemize}
\item If $p \in Q_1$ we implement an AND-gadget from
	Figure~\ref{fig:gadgets} (left):
\\
		let $(p_1X_1), (p_1'X_1), (p_2X_2), (p_2'X_2) \in \bar Q$ be fresh control states,
  and add rules 
		\begin{itemize}		
			\item
				$p X \ctran{} (p_1X_1) X \mid (p_2X_2) X$
            and $p' X \ctran{} (p_1'X_1) X \mid (p_2'X_2) X$,
\item
   $(p_1X_1) X \ctran{\ai} p_1 X_1$ and
   $(p_1'X_1) X \ctran{\ai} p_1' X_1$,
\item				
   $(p_2X_2) X \ctran{} p_2 X_2 \mid z X$ and
   $(p_2'X_2) X \ctran{} p_2' X_2 \mid z X$.
		\end{itemize}
				Here, the transitions to~$z X$ serve to implement the condition $t_1\not\sim t'_2$ from Proposition~\ref{prop:ANDgadget}.

		\end{itemize}
\end{itemize}
An induction argument now easily establishes that
 $p_0 X_0 \sim p_0' X_0$ holds in~$\bar\Delta$ if and only if
 Player~$0$ has a winning strategy in the game in~$\Delta$.

We note that the same reduction works for non-probabilistic vPDA,
 if the probabilistic branching is replaced by non-deterministic
 branching; we thus get the mentioned alternative proof of
 EXPTIME-hardness for standard vPDA.
\end{proof}

\section{Conclusion}\label{sec:conclusion}
There are a number of variants of standard (non-deterministic)
pushdown automata for which the problem of checking bisimilarity is
elementary.  For three of the most prominent such
classes---one-counter automata, visibly pushdown automata, and basic
process algebra---we have shown that checking bisimilarity for the
probabilistic extensions incurs no cost in computational
complexity over the standard case.  More precisely, for the respective
probabilistic extensions of these three models we recover the same
complexity upper bounds for checking bisimilarity as in the standard
case.  Thus the message of this paper is that adding probability comes
at no extra cost for checking bisimilarity of pushdown automata.

At a technical level, the paper is constructed around a simple
equivalence-preserving transformation that eliminates probabilistic
transitions.  However since this transformation incurs an exponential
blow-up when used as a black-box, we have had to resort to bespoke
arguments for each subclass of pPDA in order to obtain optimal
complexity bounds from our basic underlying reduction. 
We also note that hitherto rather bespoke proofs
were used in this area~\cite{Brazdil08,FuKatoen11}.
Thus a natural
question arising from this work is whether one can identify general
conditions on a class of pPDA that enable this ``bisimilarity
reduction'' to go through without incurring an exponential blow-up.

A second main message of this paper is that checking bisimilarity for
(subclasses of) PDA is no easier in the fully probabilistic case than
in the standard case.  This contrasts with the situation for
language equivalence (e.g., deciding language equivalence of
non-deterministic finite automata is PSPACE-complete, whereas the
natural analog of language equivalence for fully probabilistic finite
automata is decidable in polynomial time~\cite{Tzeng}).  In light of
this, another interesting question is whether language
equivalence of fully probabilistic PDA (without $\varepsilon$-transitions) is decidable.
This question is currently open to the best of our knowledge
(and related to other problems in language
theory~\cite{DBLP:journals/iandc/ForejtJKW14}).

\bigskip
\noindent
{\bf Acknowledgements.\quad}
Vojt\v{e}ch Forejt was at Oxford University when most of the research was carried out.
Petr Jan\v{c}ar was at Techn. Univ. Ostrava,
supported by the grant GA\v{C}R:15-13784S (finishing in 2017)
of the Grant Agency of the Czech Rep.
Stefan Kiefer is supported by the Royal Society.

\bibliographystyle{alpha}
\bibliography{db}

\newcommand{\etalchar}[1]{$^{#1}$}
\begin{thebibliography}{BGKM13}

\bibitem[AW06]{AndovaW06}
S.~Andova and T.~Willemse.
\newblock Branching bisimulation for probabilistic systems: Characteristics and
  decidability.
\newblock {\em Theor. Comput. Sci.}, 356(3):325--355, 2006.

\bibitem[Bai96]{Baier96}
C.~Baier.
\newblock Polynomial time algorithms for testing probabilistic bisimulation and
  simulation.
\newblock In {\em CAV (Computer Aided Verification)}, volume 1102 of {\em
  LNCS}, pages 50--61. Springer, 1996.

\bibitem[BCMS01]{Burkart00}
O.~Burkart, D.~Caucal, F.~Moller, and B.~Steffen.
\newblock Verification on infinite structures.
\newblock In J.A. Bergstra, A.~Ponse, and S.A. Smolka, editors, {\em Handbook
  of Process Algebra}, pages 545--623. North-Holland, 2001.

\bibitem[BEM00]{DBLP:journals/jcss/BaierEM00}
C.~Baier, B.~Engelen, and M.E. Majster{-}Cederbaum.
\newblock Deciding bisimilarity and similarity for probabilistic processes.
\newblock {\em J. Comput. Syst. Sci.}, 60(1):187--231, 2000.

\bibitem[BGJ10]{BGJ:Concur10}
S.~B{\"o}hm, S.~G{\"o}ller, and P.~Jan{\v{c}}ar.
\newblock Bisimilarity of one-counter processes is {PSPACE}-complete.
\newblock In {\em CONCUR}, volume 6269 of {\em LNCS}, pages 177--191. Springer,
  2010.

\bibitem[BGJ14]{BGJ14}
S.~B{\"{o}}hm, S.~G{\"{o}}ller, and P.~Jan\v{c}ar.
\newblock Bisimulation equivalence and regularity for real-time one-counter
  automata.
\newblock {\em J. Comput. Syst. Sci.}, 80(4):720--743, 2014.

\bibitem[BGKM13]{BGKM13}
M.~Benedikt, S.~G{\"{o}}ller, S.~Kiefer, and A.S. Murawski.
\newblock Bisimilarity of pushdown automata is nonelementary.
\newblock In {\em LICS (Logic in Computer Science)}, pages 488--498. {IEEE}
  Computer Society, 2013.

\bibitem[BH97]{BaierHermanns97}
C.~Baier and H.~Hermanns.
\newblock Weak bisimulation for fully probabilistic processes.
\newblock In {\em CAV}, volume 1254 of {\em {LNCS}}, pages 119--130. Springer,
  1997.

\bibitem[BKS08]{Brazdil08}
T.~Br{\'a}zdil, A.~Ku\v{c}era, and O.~Stra\v{z}ovsk{\'y}.
\newblock Deciding probabilistic bisimilarity over infinite-state probabilistic
  systems.
\newblock {\em Acta Inf.}, 45(2):131--154, 2008.

\bibitem[CvBW12]{CBW12}
D.~Chen, F.~van Breugel, and J.~Worrell.
\newblock On the complexity of computing probabilistic bisimilarity.
\newblock In {\em FoSSaCS}, volume 7213 of {\em LNCS}, pages 437--451.
  Springer, 2012.

\bibitem[EKM06]{DBLP:journals/lmcs/KuceraEM06}
J.~Esparza, A.~Ku{\v{c}}era, and R.~Mayr.
\newblock Model checking probabilistic pushdown automata.
\newblock {\em Logical Methods in Computer Science}, 2(1), 2006.

\bibitem[EWY10]{EWY10}
K.~Etessami, D.~Wojtczak, and M.~Yannakakis.
\newblock Quasi-birth-death processes, tree-like {QBDs}, probabilistic
  1-counter automata, and pushdown systems.
\newblock {\em Perform. Eval.}, 67(9):837--857, 2010.

\bibitem[EY09]{DBLP:journals/jacm/EtessamiY09}
K.~Etessami and M.~Yannakakis.
\newblock Recursive {M}arkov chains, stochastic grammars, and monotone systems
  of nonlinear equations.
\newblock {\em J. {ACM}}, 56(1):1:1--1:66, 2009.

\bibitem[FJKW12]{DBLP:conf/fsttcs/ForejtJKW12}
V.~Forejt, P.~Jan\v{c}ar, S.~Kiefer, and J.~Worrell.
\newblock Bisimilarity of probabilistic pushdown automata.
\newblock In {\em {FSTTCS}}, volume~18 of {\em LIPIcs}, pages 448--460. Schloss
  Dagstuhl - Leibniz-Zentrum f\"ur Informatik, 2012.

\bibitem[FJKW14]{DBLP:journals/iandc/ForejtJKW14}
V.~Forejt, P.~Jan\v{c}ar, S.~Kiefer, and J.~Worrell.
\newblock Language equivalence of probabilistic pushdown automata.
\newblock {\em Information and Computation}, 237:1--11, 2014.

\bibitem[FK11]{FuKatoen11}
H.~Fu and J.-P. Katoen.
\newblock Deciding probabilistic simulation between probabilistic pushdown
  automata and finite-state systems.
\newblock In {\em FSTTCS}, volume~13 of {\em LIPIcs}, pages 445--456. Schloss
  Dagstuhl - Leibniz-Zentrum f\"ur Informatik, 2011.

\bibitem[Hol95]{Holzer95}
M.~Holzer.
\newblock On emptiness and counting for alternating finite automata.
\newblock In {\em DLT (Developments in Language Theory)}, pages 88--97, 1995.

\bibitem[Jan13]{Jancar12}
P.~Jan\v{c}ar.
\newblock Bisimilarity on basic process algebra is in 2-{ExpTime} (an explicit
  proof).
\newblock {\em Logical Methods in Computer Science}, 9(1), 2013.

\bibitem[Jan14a]{JancarIcalp14}
P.~Jan\v{c}ar.
\newblock Bisimulation equivalence of first-order grammars.
\newblock In {\em {ICALP} (Part II)}, volume 8573 of {\em LNCS}, pages
  232--243. Springer, 2014.

\bibitem[Jan14b]{JancarFossacs14}
P.~Jan\v{c}ar.
\newblock Equivalences of pushdown systems are hard.
\newblock In {\em FoSSaCS}, volume 8412 of {\em LNCS}, pages 1--28. Springer,
  2014.

\bibitem[JS07]{JancarSawaAFA:2007}
P.~Jan{\v{c}}ar and Z.~Sawa.
\newblock A note on emptiness for alternating finite automata with a one-letter
  alphabet.
\newblock {\em Inf. Process. Lett.}, 104(5):164--167, 2007.

\bibitem[Kie13]{Kiefer12}
S.~Kiefer.
\newblock {BPA} bisimilarity is {EXPTIME}-hard.
\newblock {\em Inf. Process. Lett.}, 113(4):101--106, 2013.

\bibitem[KM02]{KuceraM02}
A.~Ku\v{c}era and R.~Mayr.
\newblock On the complexity of semantic equivalences for pushdown automata and
  {BPA}.
\newblock In {\em MFCS}, volume 2420 of {\em LNCS}, pages 433--445. Springer,
  2002.

\bibitem[KM10]{KuceraM10}
A.~Ku\v{c}era and R.~Mayr.
\newblock On the complexity of checking semantic equivalences between pushdown
  processes and finite-state processes.
\newblock {\em Information and Computation}, 208(7):772--796, 2010.

\bibitem[May03]{DBLP:conf/icalp/Mayr03}
R.~Mayr.
\newblock Undecidability of weak bisimulation equivalence for 1-counter
  processes.
\newblock In {\em {ICALP}}, volume 2719 of {\em {LNCS}}, pages 570--583.
  Springer, 2003.

\bibitem[PLS00]{PhilippouLS00}
A.~Philippou, I.~Lee, and O.~Sokolsky.
\newblock Weak bisimulation for probabilistic systems.
\newblock In {\em {CONCUR}}, volume 1877 of {\em {LNCS}}, pages 334--349.
  Springer, 2000.

\bibitem[S{\'e}n05]{Senizergues05}
G.~S{\'e}nizergues.
\newblock The bisimulation problem for equational graphs of finite out-degree.
\newblock {\em SIAM J. Comput.}, 34(5):1025--1106, 2005.

\bibitem[SL94]{SL94}
R.~Segala and N.~A. Lynch.
\newblock Probabilistic simulations for probabilistic processes.
\newblock In {\em CONCUR}, volume 836 of {\em LNCS}, pages 481--496. Springer,
  1994.

\bibitem[Srb04]{Srba:roadmap:04}
J.~Srba.
\newblock {\em Roadmap of Infinite Results}, volume Vol 2: Formal Models and
  Semantics.
\newblock World Scientific Publishing Co., 2004.
\newblock (Cf. \verb=http://people.cs.aau.dk/~srba/roadmap/=).

\bibitem[Srb09]{SrbaVisiblyPDA:2009}
J.~Srba.
\newblock Beyond language equivalence on visibly pushdown automata.
\newblock {\em Logical Methods in Computer Science}, 5(1):2, 2009.

\bibitem[Tze92]{Tzeng}
W.~Tzeng.
\newblock A polynomial-time algorithm for the equivalence of probabilistic
  automata.
\newblock {\em SIAM J.\ Comput.}, 21(2):216--227, 1992.

\bibitem[Wal01]{Walukiewicz2001}
I.~Walukiewicz.
\newblock Pushdown processes: Games and model-checking.
\newblock {\em Information and Computation}, 164(2):234--263, 2001.

\bibitem[ZYS{\etalchar{+}}18]{Zhang2018}
L.~Zhang, P.~Yang, L.~Song, H.~Hermanns, C.~Eisentraut, D.N. Jansen, and J.C.
  Godskesen.
\newblock Probabilistic bisimulation for realistic schedulers.
\newblock {\em Acta Informatica}, 55(6):461--488, 2018.

\end{thebibliography}

\end{document}